\documentclass[lettersize,journal]{IEEEtran}

\usepackage{amsmath,amsfonts}
\usepackage{amsthm}
\usepackage{algorithmic}
\usepackage{algorithm}
\usepackage{array}
\usepackage[caption=false,font=footnotesize]{subfig}
\usepackage{textcomp}
\usepackage{stfloats}
\usepackage{url}
\usepackage{verbatim}
\usepackage{graphicx}
\usepackage{cite}
\usepackage{float}
\usepackage{booktabs}
\usepackage[table]{xcolor}
\usepackage{multirow}
\usepackage{threeparttable}
\usepackage{makecell}

\newtheorem{theorem}{Theorem}
\newtheorem{assumption}{Assumption}
\newtheorem{lemma}[theorem]{Lemma}
\newtheorem{proposition}[theorem]{Proposition}

\begin{document}

\title{P-GADMM: Parallel Group-Based ADMM for Asynchronous Optimization in Heterogeneous Edge Networks}

\author{Gaiguo Wei,
        Qingying Zhang,
        Heqiang Wang,
        Yu Zhang,
        and Xiaoxiong Zhong%
\thanks{This work was supported by Peng Cheng Laboratory Project under Grant PCL2025A13.}
\thanks{Gaiguo Wei and Qingying Zhang are with the Southern University of Science and Technology, Shenzhen 518055, China, and also with Pengcheng Laboratory, Shenzhen 518000, China (e-mail: 12447019@mail.sustech.edu.cn; 12547037@mail.sustech.edu.cn).}%
\thanks{Heqiang Wang and Xiaoxiong Zhong are with Pengcheng Laboratory, Shenzhen 518000, China (e-mail: wanghq02@pcl.ac.cn; xixzhong@gmail.com).}%
\thanks{Yu Zhang is with the Department of Computer Science and Engineering, Southern University of Science and Technology, Shenzhen 518055, China (e-mail: zhangy7@sustech.edu.cn).} 
\thanks{Corresponding authors: Yu Zhang and Xiaoxiong Zhong.}}
\maketitle

\begin{abstract}
The Alternating Direction Method of Multipliers (ADMM) is widely used for distributed optimization, but its synchronous implementation can suffer from efficiency loss in heterogeneous edge networks, where fast clients or groups need to wait for slower ones before global updates can be completed. Existing group-based ADMM methods reduce communication overhead through grouping, but their grouping rules usually focus on data similarity or network topology and do not explicitly account for computation heterogeneity. To address this issue, this paper proposes Parallel Group-Based ADMM (P-GADMM) for distributed optimization in heterogeneous edge networks. P-GADMM forms computation-aware edge groups according to client computational capabilities and local data sizes, which reduces training-speed variation within each group. It further combines edge-level aggregation with bounded asynchronous coordination at the cloud, allowing active groups to participate in global updates without waiting for slower groups while controlling stale group information through a delay threshold. For strongly convex group objectives, we establish convergence guarantees for an idealized form of P-GADMM under bounded group-level staleness, showing a time-averaged convergence behavior up to a staleness-induced asymptotic error neighborhood. Experiments show that P-GADMM reduces wall-clock training time compared with representative baselines while maintaining comparable final accuracy.
\end{abstract}

\begin{IEEEkeywords}
Distributed optimization, alternating direction method of multipliers (ADMM), edge learning, heterogeneous edge networks, bounded-delay asynchronous optimization.
\end{IEEEkeywords}

\section{Introduction}
\label{sec:introduction}

\IEEEPARstart{D}{istributed} optimization has been widely used for machine learning in large-scale distributed systems \cite{Nedic2009, Boyd2011, Duchi2012}. In edge networks, data are generated and stored at many geographically distributed clients, and model training is usually carried out through local computation together with global coordination under either a parameter server architecture or a consensus architecture \cite{Dean2012, Tsianos2012, Duan2023}. As these systems continue to expand, communication overhead becomes a major bottleneck. Frequent transmission of raw data or global model updates consumes bandwidth, increases latency, and raises privacy and regulatory concerns. For this reason, reducing communication cost has become an important issue in distributed learning \cite{Zhang2013, McMahan2017, Kairouz2021, Elgabli2020GADMM}. This has led to collaborative learning frameworks in which raw data remain local and only model information is exchanged.

Hierarchical cloud-edge-client architectures provide a natural structure for collaborative learning in edge networks. The client layer performs local computation, the edge layer coordinates groups of nearby or logically related clients, and the cloud layer maintains the global model. From an optimization viewpoint, the global objective is decomposed into local objectives handled by clients, while the cloud server performs global aggregation and maintains model consistency. A major difficulty in this setting is statistical heterogeneity. Local datasets are usually collected under different environments and operating conditions, so the overall data distribution is typically non-IID. As a result, local updates may deviate from the global optimum and introduce aggregation bias, which slows convergence and degrades the final model \cite{Li2020FedProx}.

The Alternating Direction Method of Multipliers (ADMM) serves as a fundamental framework for solving distributed optimization problems with consensus constraints \cite{Boyd2011}. It decomposes the global objective into local subproblems and enforces consistency through an augmented Lagrangian. Because of this structure, ADMM can balance local optimization and global consensus in a principled way. By explicitly coupling local variables through the consensus constraint, ADMM can help control the mismatch among local models and maintain stable convergence even when local data distributions differ across clients \cite{He2025}. This makes ADMM attractive for heterogeneous distributed learning scenarios.

However, applying ADMM in heterogeneous edge networks remains difficult. Clients may differ substantially in processing capability, available resources, and communication conditions because of hardware diversity, changing resource availability, and competing background workloads. In the standard distributed implementation of ADMM, global updates are usually performed synchronously. The central server must wait until all clients, including stragglers, finish their local computation before the next update starts. During this process, faster clients remain idle while slower clients continue computing, so the overall training speed is limited by the slowest client. This causes inefficient resource usage, extra waiting time, and lower system throughput. As a result, standard ADMM is difficult to use in edge networks with strict latency requirements \cite{He_StragglerResilient, Xu2024}. More broadly, this reflects a basic challenge in distributed optimization: coordination overhead should be reduced without weakening convergence reliability \cite{Nedic2018, Elgabli2020GADMM}.

Although grouping and asynchronous optimization have been studied to alleviate synchronization delay, they do not fully address the computation heterogeneity in edge networks. Grouping methods usually do not explicitly consider differences in client processing speed, while asynchronous methods need additional control over stale information. Therefore, an effective method should improve training efficiency while preserving the stability needed for reliable convergence.

\begin{figure}[ht!]
\centering
\includegraphics[width=0.48\linewidth]{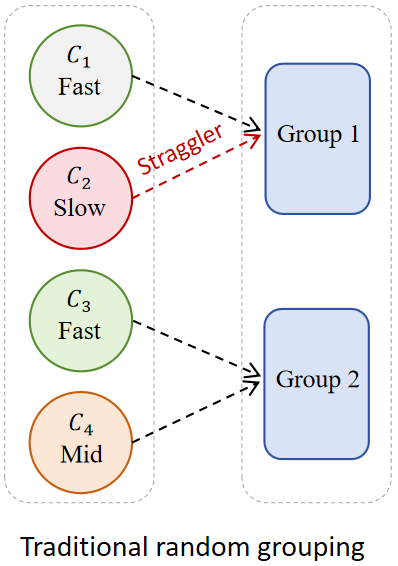}
\hfill
\includegraphics[width=0.48\linewidth]{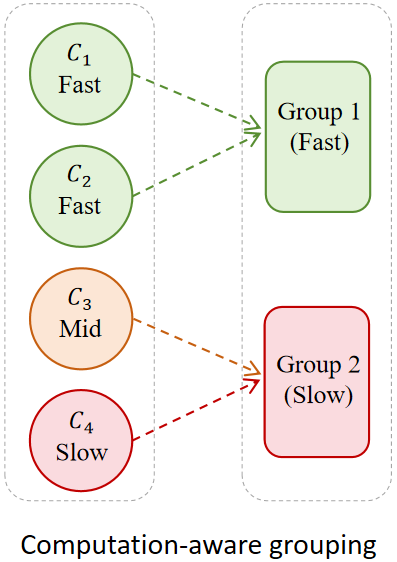}
\caption{Comparison of grouping mechanisms. In traditional random grouping (left), fast clients may wait for slow ones, which creates a straggler bottleneck. In the proposed strategy (right), clients are grouped according to computational capability and local data size, which reduces intra-group speed variation and synchronization delay.}
\label{fig_1}
\end{figure}

In this paper, we propose parallel group-based ADMM (P-GADMM) for heterogeneous edge networks. Unlike conventional grouping methods determined by topology or data characteristics, P-GADMM organizes clients according to computational capability and local data size, as illustrated in Fig.~\ref{fig_1}. It also combines grouping with bounded asynchronous coordination, so faster groups can continue updating without being blocked by slower ones. The main contributions of this work are summarized as follows:

\begin{enumerate}
    \item We propose P-GADMM for distributed optimization over a hierarchical cloud-edge-client architecture. The algorithm constructs computation-aware edge groups based on each client's computational capability and local data size, and combines aggregation at the edge layer with bounded-delay asynchronous coordination at the cloud. This design can reduce synchronization delay caused by heterogeneous client computation while keeping stale group information under an explicit delay bound.
    
   \item We analyze the convergence of an idealized version of P-GADMM under strongly convex objectives and bounded staleness. The analysis shows how delayed global model information introduces a bounded perturbation into the group optimality condition. Using a Lyapunov function, we derive a descent inequality and obtain a time-averaged bound with an $\mathcal{O}(1/T)$ decay term up to an $\mathcal{O}(\tau_{\max}^2)$ error neighborhood.
    
  \item We evaluate the performance of P-GADMM in heterogeneous edge networks and compare it with other baselines under different heterogeneity settings. The results show that P-GADMM improves convergence efficiency by reducing wall-clock training time and waiting overhead, while maintaining comparable final accuracy. 
\end{enumerate}

The rest of this paper is organized as follows. Section~\ref{sec:related_work} reviews related work. Section~\ref{sec:system_model} presents the system architecture and problem formulation. Section~\ref{sec:proposed_algorithm} describes the proposed P-GADMM algorithm. Section~\ref{sec:convergence} gives the convergence analysis. Section~\ref{sec:experiments} presents the experimental results. Section~\ref{sec:conclusion} concludes the paper.

\section{Related Work}
\label{sec:related_work}

\subsection{Distributed Optimization in Edge Networks}

Edge data are usually generated and stored across a large number of clients. This has led to distributed optimization methods in which clients perform local computation and exchange model information with servers, rather than uploading raw data. To reduce the communication between clients and the cloud, a common approach is to adopt hierarchical architectures with intermediate aggregation servers. For example, hierarchical frameworks \cite{Liu2020HFL} and heterogeneous edge aggregation schemes \cite{Abad2020} reduce direct communication between clients and the cloud. Another line of work uses client selection to reduce synchronization delay \cite{Shi2023, Tian2022}. In these methods, only a subset of clients participates in each iteration, so the waiting time caused by slow clients can be reduced.

Although these methods improve training efficiency, they still have limitations in heterogeneous edge networks. When slow clients are repeatedly excluded, the global model may become biased, especially if these clients hold distinctive local data distributions. Some studies further consider resource and data heterogeneity by adjusting local computation settings, such as adaptive local epoch sizes \cite{Yao2023MSN}. However, these methods mainly focus on communication and scheduling. In edge networks with strong heterogeneity, a more basic issue is how to preserve reliable global consensus while maintaining system efficiency. This issue motivates the use of optimization methods designed for consensus-constrained problems.

\subsection{ADMM and Grouped Distributed Optimization}

ADMM is a standard method for distributed optimization with consensus constraints. Because it decomposes a global problem into local subproblems and enforces consistency through dual variables, it has been used in many distributed learning and network optimization problems. In edge networks with limited resources, existing studies have extended ADMM in several directions. Representative examples include stochastic formulations with random data sampling \cite{Ouyang_StochasticADMM} and accelerated solvers for large-scale datasets \cite{Wang_FastADMM}. ADMM has also been used in specific distributed learning tasks, including distributed linear classification \cite{Zhang_EfficientDist}. In addition, some studies consider inexact splitting, where local subproblems are solved approximately rather than exactly, in order to reduce the computation burden on local clients \cite{Zhou_InexactADMM, Kant_ThreeOpADMM}. This usually improves local efficiency, but it may also increase the number of communication rounds before convergence.

Another way to reduce global communication overhead is to organize clients into groups. Group-based ADMM (GADMM) \cite{Wang2017GADMM} is a representative example, where clients with similar local variables are grouped and coordination is performed at the group level. Later studies extended this idea to different communication structures. Ring-based GADMM \cite{Elgabli2020GADMM} and quantized variants \cite{Elgabli2020QGADMM} restrict parameter exchange to neighboring clients. Layer-wise training schemes \cite{Elgabli2020LFGADMM}, grouped ring structures \cite{Huang2021GRADMM}, and bipartite graph models \cite{BenIssaid2022} further divide the global coordination process into smaller units. These methods improve communication efficiency through structured coordination.

However, most grouped ADMM methods still rely on strict synchronization. Moreover, grouping is usually determined by topology or data characteristics, rather than by differences in computing speed. As a result, computational imbalance may still remain within a group, and one slow client may still delay the entire group. Existing studies therefore do not fully combine the consensus structure of ADMM with grouping strategies that explicitly reflect system heterogeneity.

\subsection{Asynchronous Optimization and Straggler Mitigation}

Asynchronous updates are commonly used to reduce the waiting time introduced by synchronous coordination. Early work by Zhang et al. \cite{Zhang2014Async} and Chang et al. \cite{Chang2016Async} showed that distributed optimization can still converge under bounded delay, even when stale global variables are used. Based on this idea, later studies further considered distributed momentum methods \cite{Pond2026}, straggler-resilient consensus algorithms \cite{He_StragglerResilient}, and online learning systems \cite{Chen2020Async}. Other approaches mitigate stragglers at the system level. Gradient coding \cite{Tandon2017} mitigates straggler effects by replicating data partitions across clients, so that the server can recover the exact gradient from the earliest responses. Caching-based methods store asynchronous updates in the cloud and classify clients according to historical response times \cite{Wu2021SAFA}. More recent work has also considered hierarchical ADMM architectures for dynamic node failures \cite{Azimi2025} and asynchronous grouping methods for shared-memory multi-core clusters \cite{Zhou2018HAG}.

These methods still do not fully match the requirements of geographically distributed edge networks. Shared-memory methods \cite{Zhou2018HAG} rely on high-speed internal data buses and are not suitable for clients connected through wide-area networks. Redundancy-based methods such as gradient coding \cite{Tandon2017} introduce extra computation because overlapping data partitions must be processed repeatedly, which is difficult to sustain on energy-constrained edge devices. Fully asynchronous methods also provide only limited control over stale updates, which makes stable consensus harder to maintain in practice.

These observations suggest that an effective method for heterogeneous edge networks should combine the consensus structure of ADMM, a grouping strategy that reflects computational heterogeneity, and asynchronous communication with explicit delay control. The proposed algorithm is developed along this direction.

\section{System Model and Problem Formulation}
\label{sec:system_model}

In this section, we introduce the system model and the optimization problem considered. We consider a heterogeneous edge network with a three-layer architecture, namely, a client layer, an edge group layer, and a cloud server layer. This architecture is adopted to reflect the practical structure of edge networks, where local computation is carried out at clients, intermediate coordination is handled by edge groups, and global aggregation is performed by the cloud server. Since clients may differ significantly in computational capability, the group layer is used to reduce the impact of slow clients on the overall training process.

\begin{table}[!t]
\centering
\renewcommand{\arraystretch}{1.05}
\caption{Summary of Main Notations}
\label{tab:notations}
\footnotesize
\begin{tabular}{>{\raggedright\arraybackslash}p{0.9cm} >{\raggedright\arraybackslash}p{6.5cm}}
\toprule
\textbf{Notation} & \textbf{Description} \\
\midrule
$M$, $G$ & numbers of clients and groups \\
$\mathcal{G}_g$ & set of clients in group $g$ \\
$m_g$ & number of clients in group $g$ \\
$d$ & model dimension \\
$C_i$ & computational capability of client $i$ \\
$D_i$ & local dataset of client $i$ \\
$|D_i|$ & size of local dataset $D_i$ \\
$T_i$ & estimated training latency of client $i$ \\
$w^k$ & global model at iteration $k$ \\
$w_g^k$ & group model of group $g$ at iteration $k$ \\
$w_g^{k,i}$ & local model of client $i$ in group $g$ at iteration $k$ \\
$\bar{w}_g^k$ & empirical average of local client updates in group $g$ \\
$F(w)$ & global objective function \\
$f_i(w)$ & local objective function of client $i$ \\
$F_g(w)$ & group objective function of group $g$ \\
$\xi_i$ & random data sample of client $i$ \\
$\lambda_g^k$ & dual variable of group $g$ at iteration $k$ \\
$\rho$ & penalty parameter \\
$\mathcal{L}_\rho$ & augmented Lagrangian \\
$k$ & global iteration index \\
$T$ & total number of iterations \\
$d_g(k)$ & index of the global model used by group $g$ at iteration $k$ \\
$\tau_g(k)$ & staleness of group $g$ at iteration $k$ \\
$\tau_{\max}$ & maximum allowable staleness \\
$\zeta$ & learning rate in local SGD \\
$A(k)$ & active set of groups at iteration $k$ \\
\bottomrule
\end{tabular}
\end{table}

\subsection{Computation-Aware Grouping Strategy}
\label{subsec:grouping_strategy}

We first introduce the grouping strategy. In conventional hierarchical methods, clients are often grouped randomly or according to topology. Such strategies do not directly consider differences in computational capability, so fast clients and stragglers may still be placed in the same group. Under synchronous group aggregation, the update speed of that group is then determined by its slowest client. This leads to unnecessary waiting time for faster clients and weakens the advantage of parallel execution.

To reduce this effect, we adopt a static computation-aware grouping strategy. The basic idea is simple: clients with similar processing speeds are assigned to the same group. In this way, the variation of training time within each group can be reduced.

Let \( C_i \) denote the computational capability of client \( i \), such as CPU cycles per second. Let \( D_i \) denote the local dataset of client \( i \), and let \( |D_i| \) denote its size. We define the estimated local training latency per epoch as
\[
T_i = \frac{|D_i|}{C_i}.
\]
Based on \(\{T_i\}_{i=1}^M\), the full set of \(M\) clients is partitioned into \(G\) disjoint groups \(\{\mathcal{G}_1,\dots,\mathcal{G}_G\}\).

The grouping procedure has three steps. First, the cloud server collects \(C_i\) and \(D_i\) from all participating clients and computes the estimated training latency \(T_i\). Second, the clients are sorted in ascending order of \(T_i\), which gives an ordered sequence \(\{i_1,i_2,\dots,i_M\}\) satisfying
\[
T_{i_1} \le T_{i_2} \le \cdots \le T_{i_M}.
\]
Third, this ordered sequence is divided into \(G\) contiguous segments of nearly equal size, and each segment forms one group.

With this grouping strategy, clients in the same group have closer training speeds. As a result, intra-group synchronization delay can be reduced, and the efficiency of group updates can be improved.

\subsection{Distributed Optimization Formulation}
\label{subsec:problem_formulation}

Based on the grouping configuration \(\{\mathcal{G}_g\}_{g=1}^G\), we now define the distributed optimization problem. The total number of participating clients is
\[
M = \sum_{g=1}^G |\mathcal{G}_g|,
\]
where \(m_g = |\mathcal{G}_g|\) denotes the number of clients in group \(g\).

Let \(w \in \mathbb{R}^d\) denote the global model parameter, and let \(w_g \in \mathbb{R}^d\) denote the group-level model parameter associated with group \(g\). For each client \(i \in \mathcal{G}_g\), let \(w_g^{k,i} \in \mathbb{R}^d\) denote the corresponding local model parameter at iteration \(k\).

For client \(i\), the local objective is defined as the expected loss over its local data distribution:
\begin{equation}
f_i(w) = \mathbb{E}_{\xi_i \in D_i}\left[f_i\left(w;\xi_i\right)\right],
\label{eq:local_loss}
\end{equation}
where \(\xi_i\) is a random sample drawn from the local dataset \(D_i\).

The global objective is defined as the average of all local objectives:
\begin{equation}
F(w) = \frac{1}{M}\sum_{i=1}^M f_i(w)
= \sum_{g=1}^G \frac{m_g}{M}F_g(w),
\label{eq:global_loss}
\end{equation}
where
\[
F_g(w) = \frac{1}{m_g}\sum_{i\in \mathcal{G}_g} f_i(w)
\]
denotes the aggregated objective of group \(g\).

To apply ADMM in this hierarchical model, we introduce group-level copies \(\{w_g\}_{g=1}^G\) and enforce global consensus at the cloud server. The original optimization problem can then be written as
\begin{equation}
\begin{aligned}
\underset{w, \{w_g\}_{g=1}^G}{\text{minimize}} \quad &
\sum_{g=1}^G \frac{m_g}{M}F_g(w_g), \\
\text{subject to} \quad &
w_g = w,\quad \forall g = 1,\dots,G.
\end{aligned}
\label{eq:dist_opt}
\end{equation}

\begin{figure}[t]
    \centering
    \includegraphics[width=1.0\linewidth]{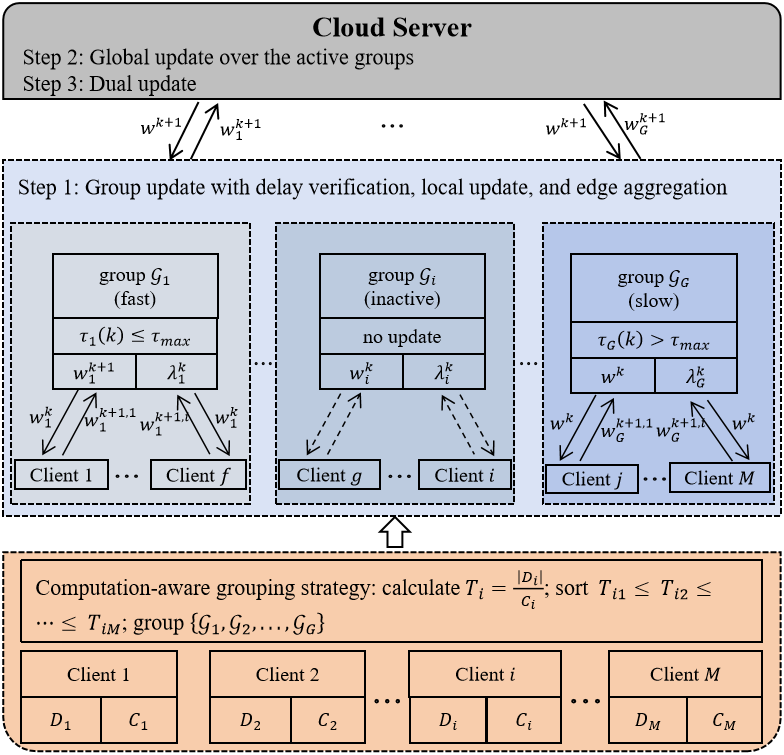}
   \caption{Workflow of one P-GADMM iteration in the hierarchical cloud-edge-client architecture. Clients are partitioned into edge groups according to the estimated local training latency \(T_i=|D_i|/C_i\). Active groups perform delay verification, local computation, and edge aggregation, whereas inactive groups keep variables unchanged. The cloud server then updates the global and dual variables using the current group variables.}
    \label{fig_2}
\end{figure}

\section{Proposed P-GADMM Algorithm}
\label{sec:proposed_algorithm}

In this section, we introduce the proposed P-GADMM algorithm. To solve the constrained optimization problem in \eqref{eq:dist_opt} without strict global synchronization, P-GADMM combines synchronized computation within each group and asynchronous coordination across groups. In the implemented algorithm, each active group performs local training in parallel within the group and then sends its aggregated update to the cloud server. For the convergence analysis, this practical group update is represented by an idealized ADMM subproblem.

With dual variables \(\lambda_g \in \mathbb{R}^d\) and penalty parameter \(\rho>0\), the augmented Lagrangian is given by
\begin{equation}
\mathcal{L}_\rho
=
\sum_{g=1}^G
\left[
\frac{m_g}{M}F_g(w_g)
+
\langle \lambda_g, w_g-w\rangle
+
\frac{\rho}{2}\|w_g-w\|^2
\right].
\label{eq:lagrangian}
\end{equation}

At iteration \(k+1\), P-GADMM has three main steps: group update, global update, and dual update.

\subsection{Step 1: Group Update}

At iteration \(k\), let \(A(k)\subseteq \{1,\dots,G\}\) denote the active set, which contains the groups whose updates are received by the cloud server and used to construct \(w^{k+1}\). Only groups in \(A(k)\) perform group update at iteration \(k\). For \(g\notin A(k)\), group \(g\) does not contribute a new update to the current global update, and its previous group variables remain unchanged. The group partition \(\{\mathcal{G}_1,\dots,\mathcal{G}_G\}\) is fixed after computation-aware grouping, whereas \(A(k)\) may vary across iterations.

\subsubsection{Delay Verification}

For each active group \(g\in A(k)\), the edge server first verifies the staleness of its cached global reference before local computation. Let \(d_g(k)\) denote the iteration index of the latest global model cached at the edge server and used as the reference for group \(g\) at iteration \(k\). The corresponding staleness is defined as
\begin{equation}
\tau_g(k)=k-d_g(k).
\label{eq:staleness}
\end{equation}

Here, \(w_g^k\) denotes the cached reference model maintained at the edge server of group \(g\). If \(\tau_g(k)\le \tau_{\max}\), the cached reference is used directly. Otherwise, the edge server updates the cached reference using the current global model from the cloud server:
\begin{equation}
w_g^k \leftarrow w^k,
\qquad
d_g(k)\leftarrow k.
\label{eq:delay_reset}
\end{equation}
After verification, local computation in group \(g\) starts from the reference model \(w_g^k\). This step ensures that every active group performs local computation using a global reference whose staleness is bounded by \(\tau_{\max}\).

\subsubsection{Client-Side Local Update}

After delay verification, each active group \(g\in A(k)\) performs synchronized local training. Specifically, client \(i\in\mathcal{G}_g\) starts from the current group model \(w_g^k\) and performs a local SGD step:
\begin{equation}
w_g^{k+1,i}
=
w_g^k-\zeta \nabla f_i(w_g^k;\xi_i),
\label{eq:client_local_update}
\end{equation}
where \(\zeta>0\) is the learning rate and \(\xi_i\) denotes a mini-batch sampled from the local dataset. If multiple local steps are used, \eqref{eq:client_local_update} is applied repeatedly.

\subsubsection{Edge Aggregation}

After local training, the edge server aggregates the client models in group \(g\) as
\begin{equation}
\bar{w}_g^{k+1}
=
\frac{1}{m_g}\sum_{i\in\mathcal{G}_g} w_g^{k+1,i}.
\label{eq:empirical_average}
\end{equation}
For active groups, the practical group update sent to the cloud server is
\begin{equation}
w_g^{k+1}=\bar{w}_g^{k+1},
\qquad
g\in A(k).
\label{eq:active_group_update}
\end{equation}
For inactive groups, no new local computation is triggered in the current round, and the previous group variable is carried forward:
\begin{equation}
w_g^{k+1}=w_g^k,
\qquad
g\notin A(k).
\label{eq:inactive_group}
\end{equation}

\subsubsection{Analytical Model for Convergence Analysis}

For the convergence analysis, we introduce an idealized group update that preserves the ADMM consensus structure. For each active group \(g\in A(k)\), the group variable is defined as the exact solution of the following delayed ADMM subproblem:
\begin{equation}
\begin{aligned}
w_g^{k+1}
&=
\arg\min_{w_g}
\Bigg[
\frac{m_g}{M}F_g(w_g)
+
\big\langle
\lambda_g^{d_g(k)},\, w_g-w^{d_g(k)}
\big\rangle \\
&\qquad\qquad
+
\frac{\rho}{2}\|w_g-w^{d_g(k)}\|^2
\Bigg].
\end{aligned}
\label{eq:intra_group_aggregation}
\end{equation}
This idealized subproblem provides the analytical reference used in the convergence proof.

\begin{algorithm}[t]
\caption{Parallel Group-Based ADMM (P-GADMM)}
\label{alg:pgadmm}
\begin{algorithmic}[1]
\REQUIRE Datasets $\mathbf{D}=\{D_1,\dots,D_M\}$, learning rate $\zeta$, penalty parameter $\rho$, delay bound $\tau_{\max}$, total number of global iterations $T$, client capabilities $\{C_i\}$
\ENSURE Final global model $w^T$

\STATE \textbf{Initialization:} Set $k\leftarrow 0$. Initialize $w^0$, $w_g^0$, and $\lambda_g^0\leftarrow \mathbf{0}$ for all groups.
\STATE \textbf{Grouping:} Partition the $M$ clients into $G$ disjoint groups $\{\mathcal{G}_g\}$ according to the estimated training latency.
\STATE Set $d_g(0)\leftarrow 0$ for all groups.

\WHILE{$k<T$ and global convergence is not achieved}

    \STATE Determine the active set $A(k)$.

    \STATE \textbf{$\triangleright$ Step 1: Group Update}
    \FOR{each group $g\in A(k)$ \textbf{in parallel}}
        \STATE Compute the staleness index $\tau_g(k)\leftarrow k-d_g(k)$.
        \IF{$\tau_g(k)>\tau_{\max}$}
            \STATE Refresh the group model according to \eqref{eq:delay_reset}.
        \ENDIF
        \FOR{each client $i\in\mathcal{G}_g$ \textbf{in parallel}}
            \STATE Update the local model according to \eqref{eq:client_local_update}.
        \ENDFOR
        \STATE Aggregate the local models according to \eqref{eq:empirical_average}.
        \STATE Set $w_g^{k+1}$ according to \eqref{eq:active_group_update}.
    \ENDFOR

    \FOR{each group $g\notin A(k)$}
        \STATE Set $w_g^{k+1}$ according to \eqref{eq:inactive_group}.
    \ENDFOR

    \STATE \textbf{$\triangleright$ Step 2: Global Update}
    \STATE Update $w^{k+1}$ according to \eqref{eq:global_update_closedform}.

    \STATE \textbf{$\triangleright$ Step 3: Dual Update}
    \FOR{each group $g\in\{1,\dots,G\}$}
        \STATE Update $\lambda_g^{k+1}$ according to \eqref{eq:dual_update}.
    \ENDFOR

    \STATE $k\leftarrow k+1$.
\ENDWHILE
\end{algorithmic}
\end{algorithm}

\subsection{Step 2: Global Update}

After all group variables \(\{w_g^{k+1}\}_{g=1}^G\) are determined, the cloud server updates the global consensus variable by minimizing the augmented Lagrangian with respect to \(w\):
\begin{equation}
w^{k+1}
=
\arg\min_w
\sum_{g=1}^G
\left[
\langle \lambda_g^k,\, w_g^{k+1}-w\rangle
+
\frac{\rho}{2}\|w_g^{k+1}-w\|^2
\right].
\label{eq:global_optimization}
\end{equation}
This quadratic problem has the closed-form solution
\begin{equation}
w^{k+1}
=
\frac{1}{G}\sum_{g=1}^G
\left(
w_g^{k+1}+\frac{1}{\rho}\lambda_g^k
\right).
\label{eq:global_update_closedform}
\end{equation}

\subsection{Step 3: Dual Update}

After the global update, the cloud server updates the dual variables to penalize the discrepancy between the group variables and the new global consensus variable. For every group \(g=1,\dots,G\), the dual update is
\begin{equation}
\lambda_g^{k+1}
=
\lambda_g^k+\rho\bigl(w_g^{k+1}-w^{k+1}\bigr).
\label{eq:dual_update}
\end{equation}
This update is applied to both active and inactive groups. For inactive groups, \(w_g^{k+1}\) in \eqref{eq:dual_update} is given by the unchanged group variable defined in \eqref{eq:inactive_group}.

Algorithm~\ref{alg:pgadmm} summarizes the implemented version of P-GADMM used in the experiments. Section~\ref{sec:convergence} analyzes the idealized group update in \eqref{eq:intra_group_aggregation} as the analytical model of the implemented group step.

\section{Convergence Analysis}
\label{sec:convergence}

In this section, we analyze the idealized analytical model introduced in Section~\ref{sec:proposed_algorithm}. The purpose of this analysis is to show how bounded delay affects the convergence behavior of P-GADMM. We first derive a perturbed first-order condition for the group update. We then use this relation to construct a Lyapunov analysis and obtain a descent bound under bounded delay.

\subsection{Preliminaries}

We start with three standard assumptions.

\begin{assumption}[$\mu_g$-Strong Convexity]
\label{assu:strong_convex}
For each group \(g\), the objective function \(F_g\) is \(\mu_g\)-strongly convex. That is, there exists a constant \(\mu_g>0\) such that, for all \(x,y \in \mathrm{dom}(F_g)\),
\[
(\nabla F_g(x)-\nabla F_g(y))^{\mathrm T}(x-y)\ge \mu_g\|x-y\|^2.
\]
\end{assumption}

\begin{assumption}[Saddle Point Existence]
\label{assu:strong_duality}
There exists a saddle point
\[
\Omega^*=\bigl(w^*,\{w_g^*\}_{g=1}^G,\{\lambda_g^*\}_{g=1}^G\bigr)
\]
for the Lagrangian associated with the distributed optimization problem. This saddle point satisfies the KKT conditions, including primal consistency \(w_g^*=w^*\) for all \(g\), and dual optimality
\[
\frac{m_g}{M}\nabla F_g(w_g^*)+\lambda_g^*=0,\qquad g=1,\dots,G.
\]
\end{assumption}

\begin{assumption}[Uniformly Bounded Delay]
\label{assu:bounded_delay}
The delay
\[
\tau_g(k)=k-d_g(k)
\]
is uniformly bounded by a finite constant \(\tau_{\max}<\infty\) for all groups \(g\) and all iterations \(k\).
\end{assumption}

The optimality conditions at the saddle point will be used repeatedly in the analysis. For each group \(g\), they are
\begin{equation}
\frac{m_g}{M}\nabla F_g(w_g^*)+\lambda_g^*=0,\qquad
w_g^*=w^*,
\label{eq:kkt_local}
\end{equation}
together with the dual feasibility condition
\begin{equation}
\sum_{g=1}^G \lambda_g^*=0.
\label{eq:kkt_global_sum}
\end{equation}

\subsection{Perturbed Optimality Condition}

We first derive the first-order condition of the stale group subproblem and isolate the error introduced by delay.

\begin{lemma}[Perturbed Optimality]
\label{lemma:optimality}
For any active group \(g\in A(k)\), the update \(w_g^{k+1}\) satisfies
\begin{equation}
\frac{m_g}{M}\nabla F_g(w_g^{k+1})+\lambda_g^{k+1}+\rho(w_g^{k+1}-w^{k+1})=\xi_g^k,
\label{eq:lemma1_result}
\end{equation}
where
\begin{equation}
\xi_g^k=(\lambda_g^{k+1}-\lambda_g^{d_g(k)})+\rho(w^{d_g(k)}-w^{k+1}).
\label{eq:xi_def}
\end{equation}
\end{lemma}

\begin{IEEEproof}
For an active group \(g\in A(k)\), the analytical update is defined by
\begin{equation}
\begin{split}
w_g^{k+1}
=
\arg\min_{w_g}
\bigg[
\frac{m_g}{M}F_g(w_g)
+ \langle \lambda_g^{d_g(k)},\, w_g-w^{d_g(k)}\rangle \\
+ \frac{\rho}{2}\|w_g-w^{d_g(k)}\|^2
\bigg].
\end{split}
\end{equation}
The first-order optimality condition of this convex subproblem is
\begin{equation}
\frac{m_g}{M}\nabla F_g(w_g^{k+1})+\lambda_g^{d_g(k)}+\rho(w_g^{k+1}-w^{d_g(k)})=0.
\end{equation}
Add and subtract \(\lambda_g^{k+1}\) and \(w^{k+1}\). Then
\begin{equation}
\begin{split}
&\frac{m_g}{M}\nabla F_g(w_g^{k+1})+\lambda_g^{k+1}+\rho(w_g^{k+1}-w^{k+1}) \\
&\qquad =(\lambda_g^{k+1}-\lambda_g^{d_g(k)})+\rho(w^{d_g(k)}-w^{k+1}),
\end{split}
\end{equation}
which is exactly \eqref{eq:lemma1_result} with \(\xi_g^k\) defined in \eqref{eq:xi_def}.
\end{IEEEproof}

\subsection{Lyapunov Difference}

We define the Lyapunov function as
\begin{equation}
V^k
=
\sum_{g=1}^G\left(\frac{1}{\rho}\|\lambda_g^k-\lambda_g^*\|^2+\rho\|w^k-w^*\|^2\right).
\end{equation}
Since only the active groups update their local variables at iteration \(k\), the dual variation only appears on \(A(k)\). The primal part is associated with the shared global model \(w\). Therefore, the Lyapunov difference can be written as
\begin{equation}
\begin{split}
\Delta V^k
&=
\sum_{g\in A(k)} \frac{1}{\rho}
\Big(
\|\lambda_g^{k+1}-\lambda_g^*\|^2
-
\|\lambda_g^{k}-\lambda_g^*\|^2
\Big) \\
&\quad
+
\rho G
\left(
\|w^{k+1}-w^*\|^2-\|w^{k}-w^*\|^2
\right) \\
&=
\Delta V_{\mathrm{dual}}+\Delta V_{\mathrm{primal}}.
\end{split}
\end{equation}

We first expand the dual term. For each \(g\in A(k)\),
\begin{align}
\|\lambda_g^{k+1}-\lambda_g^*\|^2
&=
\left\|(\lambda_g^{k+1}-\lambda_g^k)+(\lambda_g^k-\lambda_g^*)\right\|^2 \nonumber\\
&=
\|\lambda_g^{k+1}-\lambda_g^k\|^2+\|\lambda_g^k-\lambda_g^*\|^2 \nonumber\\
&\quad +2\langle \lambda_g^{k+1}-\lambda_g^k,\lambda_g^k-\lambda_g^*\rangle.
\end{align}
Hence,
\begin{align}
&\|\lambda_g^{k+1}-\lambda_g^*\|^2-\|\lambda_g^k-\lambda_g^*\|^2 \nonumber\\
&=
\|\lambda_g^{k+1}-\lambda_g^k\|^2
+2\langle \lambda_g^{k+1}-\lambda_g^k,\lambda_g^k-\lambda_g^*\rangle.
\end{align}

Rewrite \(\lambda_g^k\) in the inner product as
\(\lambda_g^{k+1}-(\lambda_g^{k+1}-\lambda_g^k)\). Then
\begin{alignat}{1}
&\langle \lambda_g^{k+1}-\lambda_g^k,\lambda_g^k-\lambda_g^*\rangle \nonumber\\
&=
\langle \lambda_g^{k+1}-\lambda_g^k,(\lambda_g^{k+1}-\lambda_g^*)-(\lambda_g^{k+1}-\lambda_g^k)\rangle \nonumber\\
&=
\langle \lambda_g^{k+1}-\lambda_g^k,\lambda_g^{k+1}-\lambda_g^*\rangle
-
\|\lambda_g^{k+1}-\lambda_g^k\|^2.
\end{alignat}
Substituting this identity yields
\begin{align}
&\|\lambda_g^{k+1}-\lambda_g^*\|^2-\|\lambda_g^k-\lambda_g^*\|^2 \nonumber\\
&=
2\langle \lambda_g^{k+1}-\lambda_g^k,\lambda_g^{k+1}-\lambda_g^*\rangle
-
\|\lambda_g^{k+1}-\lambda_g^k\|^2.
\end{align}

Using the dual update \eqref{eq:dual_update}, we obtain
\begin{equation}
\Delta V_{\mathrm{dual}}^g
=
2\langle w_g^{k+1}-w^{k+1},\lambda_g^{k+1}-\lambda_g^*\rangle
-
\rho\|w_g^{k+1}-w^{k+1}\|^2.
\label{eq:dual_expansion_final}
\end{equation}

\subsection{Strong Convexity Coupling}

We next combine the perturbed optimality condition with strong convexity. From Assumption~\ref{assu:strong_convex},
\begin{equation}
\langle \nabla F_g(w_g^{k+1})-\nabla F_g(w^*),\,w_g^{k+1}-w^*\rangle
\ge
\mu_g\|w_g^{k+1}-w^*\|^2.
\end{equation}
Using Lemma~\ref{lemma:optimality}, the identity
\[
\nabla F_g(w^*)=-\frac{M}{m_g}\lambda_g^*
\]
from \eqref{eq:kkt_local}, and multiplying both sides by \(\frac{m_g}{M}\), we have
\begin{equation}
\begin{split}
&\Big\langle
\xi_g^k-\lambda_g^{k+1}-\rho(w_g^{k+1}-w^{k+1})+\lambda_g^*,\,
w_g^{k+1}-w^*
\Big\rangle \\
&\qquad \ge \frac{m_g\mu_g}{M}\|w_g^{k+1}-w^*\|^2.
\end{split}
\end{equation}
Rearranging gives
\begin{equation}
\label{eq:primal_bound}
\begin{aligned}
&\bigl\langle \lambda_g^{k+1}-\lambda_g^*,\; w_g^{k+1}-w^* \bigr\rangle
\\
&\le
- \frac{m_g\mu_g}{M}\|w_g^{k+1}-w^*\|^2
- \rho\bigl\langle w_g^{k+1}-w^{k+1},\; w_g^{k+1}-w^* \bigr\rangle
\\
&\quad + \bigl\langle \xi_g^k,\; w_g^{k+1}-w^* \bigr\rangle.
\end{aligned}
\end{equation}

Substituting \eqref{eq:primal_bound} into \eqref{eq:dual_expansion_final}, we obtain
\begin{equation}
\begin{split}
&2\langle w_g^{k+1}-w^{k+1},\lambda_g^{k+1}-\lambda_g^*\rangle
-\rho\|w_g^{k+1}-w^{k+1}\|^2 \\
&=
2\langle w_g^{k+1}-w^*,\lambda_g^{k+1}-\lambda_g^*\rangle \\
&\quad
-2\langle w^{k+1}-w^*,\lambda_g^{k+1}-\lambda_g^*\rangle
-\rho\|w_g^{k+1}-w^{k+1}\|^2 \\
&\le
-\frac{2m_g\mu_g}{M}\|w_g^{k+1}-w^*\|^2
-2\rho\langle w_g^{k+1}-w^{k+1},\,w_g^{k+1}-w^*\rangle \\
&\quad
+2\langle \xi_g^k,\,w_g^{k+1}-w^*\rangle
-\rho\|w_g^{k+1}-w^{k+1}\|^2 \\
&\quad
-2\langle w^{k+1}-w^*,\lambda_g^{k+1}-\lambda_g^*\rangle.
\end{split}
\label{eq:cross_term_expansion}
\end{equation}

\noindent\textit{Remark on the global error term.}
The last term in \eqref{eq:cross_term_expansion} is a coupling term between the global model error and the dual error. Summing over all groups removes this term:
\begin{equation*}
\begin{split}
&\sum_{g=1}^G \langle w^{k+1}-w^*,\,\lambda_g^{k+1}-\lambda_g^*\rangle \\
&\qquad =
\Bigg\langle
w^{k+1}-w^*,\,
\sum_{g=1}^G \lambda_g^{k+1}-\sum_{g=1}^G \lambda_g^*
\Bigg\rangle
=0.
\end{split}
\end{equation*}
This follows from \eqref{eq:kkt_global_sum} and the corresponding dual feasibility relation at iteration \(k+1\).

\subsection{Primal Difference and Cross Terms}

We next bound the primal part of the Lyapunov difference.

\begin{proposition}[Bound on Global Primal Difference]
The global primal variation satisfies
\begin{equation}
\Delta V_{\mathrm{primal}}
\le
\rho\sum_{g=1}^G \|w_g^{k+1}-w^*\|^2
-
\rho G \|w^k-w^*\|^2.
\end{equation}
\end{proposition}

\begin{IEEEproof}
From \eqref{eq:global_update_closedform},
\begin{equation*}
w^{k+1}
=
\frac{1}{G}\sum_{g=1}^G w_g^{k+1}
+
\frac{1}{\rho G}\sum_{g=1}^G \lambda_g^k.
\end{equation*}
Since \(\lambda_g^0=0\) and the dual variables satisfy the corresponding summation relation, the second term vanishes. Therefore,
\begin{equation*}
w^{k+1}
=
\frac{1}{G}\sum_{g=1}^G w_g^{k+1}.
\end{equation*}
Applying Jensen's inequality to \(\|\cdot\|^2\) gives
\begin{equation}
\begin{split}
\|w^{k+1}-w^*\|^2
&=
\left\|
\frac{1}{G}\sum_{g=1}^G (w_g^{k+1}-w^*)
\right\|^2 \\
&\le
\frac{1}{G}\sum_{g=1}^G \|w_g^{k+1}-w^*\|^2.
\end{split}
\end{equation}
Multiplying both sides by \(\rho G\) gives the result.
\end{IEEEproof}

We now collect the \(\rho\)-dependent terms in \eqref{eq:cross_term_expansion}. These are
\(-\rho\|w_g^{k+1}-w^{k+1}\|^2\),
the cross term
\(-2\rho\langle w_g^{k+1}-w^{k+1},\,w_g^{k+1}-w^*\rangle\),
and the positive term from \(\Delta V_{\mathrm{primal}}\).

We decompose
\begin{equation}
w_g^{k+1}-w^*
=
(w_g^{k+1}-w^{k+1})+(w^{k+1}-w^*).
\end{equation}
Substituting this into the cross term gives
\begin{equation}
\begin{alignedat}{1}
&\langle w_g^{k+1}-w^{k+1},\,w_g^{k+1}-w^*\rangle \nonumber\\
&=
\big\langle
w_g^{k+1}-w^{k+1},\,
(w_g^{k+1}-w^{k+1})+(w^{k+1}-w^*)
\big\rangle \nonumber\\
&=
\|w_g^{k+1}-w^{k+1}\|^2
+
\langle w_g^{k+1}-w^{k+1},\,w^{k+1}-w^*\rangle.
\end{alignedat}
\end{equation}
Hence,
\begin{equation}
\begin{alignedat}{1}
&-2\rho\langle w_g^{k+1}-w^{k+1},\,w_g^{k+1}-w^*\rangle \nonumber\\
&=
-2\rho\|w_g^{k+1}-w^{k+1}\|^2
-
2\rho\langle w_g^{k+1}-w^{k+1},\,w^{k+1}-w^*\rangle.
\end{alignedat}
\end{equation}

Taking expectation over the random active set \(A(k)\) removes the mixed deviation term. If each group is selected with the same probability \(p=|A(k)|/G\), then
\begin{equation}
\begin{split}
&\mathbb{E}_{A(k)}
\bigg[
\sum_{g\in A(k)}
\langle w_g^{k+1}-w^{k+1},\,w^{k+1}-w^*\rangle
\bigg] \\
&\qquad =
p\sum_{g=1}^G
\langle w_g^{k+1}-w^{k+1},\,w^{k+1}-w^*\rangle
=0.
\end{split}
\end{equation}
This holds because \(w^{k+1}\) is the average of the \(G\) group variables, and therefore
\[
\sum_{g=1}^G (w_g^{k+1}-w^{k+1})=0.
\]

Using the proposition above, the positive primal term can be upper bounded by
\[
\rho \sum_{g=1}^G \|w_g^{k+1}-w^*\|^2,
\]
and can therefore be absorbed into the negative strong convexity term.

Next, apply Young's inequality
\[
2\langle a,b\rangle \le \frac{1}{\eta}\|a\|^2+\eta\|b\|^2
\]
to the perturbation term
\(2\langle \xi_g^k,\,w_g^{k+1}-w^*\rangle\), and choose
\(\eta=\frac{m_g\mu_g}{M}\). Then
\begin{equation}
\begin{split}
2\langle \xi_g^k,\,w_g^{k+1}-w^*\rangle
\le
\frac{M}{m_g\mu_g}\|\xi_g^k\|^2
+
\frac{m_g\mu_g}{M}\|w_g^{k+1}-w^*\|^2.
\end{split}
\end{equation}

Collecting all terms yields
\begin{equation}
\begin{split}
\mathbb{E}[V^{k+1}-V^k]
&\le
-
\mathbb{E}
\bigg[
\sum_{g\in A(k)}
\left(
\frac{m_g\mu_g}{M}-\rho
\right)
\|w_g^{k+1}-w^*\|^2
\bigg] \\
&\quad
+
\mathbb{E}
\bigg[
\sum_{g\in A(k)}
\frac{M}{m_g\mu_g}\|\xi_g^k\|^2
\bigg].
\end{split}
\label{eq:descent_ineq}
\end{equation}
Therefore, if
\[
\rho<\frac{m_g\mu_g}{M},
\]
the Lyapunov function has a strict expected descent up to the delay error term.

\subsection{Explicit Bound on the Delay Error}

We now bound the delay error term
\[
\sum_{g\in A(k)} \frac{M}{m_g\mu_g}\|\xi_g^k\|^2.
\]
From \eqref{eq:xi_def},
\begin{equation}
\xi_g^k
=
(\lambda_g^{k+1}-\lambda_g^{d_g(k)})
+
\rho(w^{d_g(k)}-w^{k+1}).
\end{equation}

\subsubsection{Bounding \(\|\xi_g^k\|^2\)}

Using \(\|a+b\|^2\le 2\|a\|^2+2\|b\|^2\), we obtain
\begin{equation}
\|\xi_g^k\|^2
\le
2\|\lambda_g^{k+1}-\lambda_g^{d_g(k)}\|^2
+
2\rho^2\|w^{d_g(k)}-w^{k+1}\|^2.
\label{eq:xi_sq_bound}
\end{equation}

We first bound the primal delay term. By telescoping,
\begin{equation}
w^{k+1}-w^{d_g(k)}
=
\sum_{j=d_g(k)}^k (w^{j+1}-w^j).
\end{equation}
Since the number of summands is at most \(\tau_{\max}+1\), the discrete Cauchy--Schwarz inequality gives
\begin{equation}
\left\|
\sum_{j=d_g(k)}^k (w^{j+1}-w^j)
\right\|^2
\le
(\tau_g(k)+1)\sum_{j=d_g(k)}^k \|w^{j+1}-w^j\|^2.
\end{equation}
Because the iterates remain in a bounded region, there exists a constant \(D>0\) such that
\(\|w^{j+1}-w^j\|^2\le D\). Hence,
\begin{align}
\|w^{k+1}-w^{d_g(k)}\|^2
&\le
(\tau_{\max}+1)\sum_{j=k-\tau_{\max}}^k D \nonumber\\
&\le
(\tau_{\max}+1)^2D.
\end{align}
Therefore, there exists a constant \(C_w>0\) such that
\begin{equation}
\|w^{k+1}-w^{d_g(k)}\|^2
\le
C_w\tau_{\max}^2.
\end{equation}

The same argument applies to the dual term. Using the unified dual update \eqref{eq:dual_update}, there exists a constant \(D_\lambda>0\) such that
\[
\|\lambda_g^{j+1}-\lambda_g^j\|^2 \le D_\lambda.
\]
Therefore,
\begin{equation}
\|\lambda_g^{k+1}-\lambda_g^{d_g(k)}\|^2
\le
C_\lambda\tau_{\max}^2,
\end{equation}
for some constant \(C_\lambda>0\).

Substituting these bounds into \eqref{eq:xi_sq_bound} gives
\begin{equation}
\begin{split}
\|\xi_g^k\|^2
&\le
2C_\lambda\tau_{\max}^2
+
2\rho^2C_w\tau_{\max}^2 \\
&=
2(C_\lambda+\rho^2C_w)\tau_{\max}^2.
\end{split}
\end{equation}

\subsubsection{Bounding the Weighted Sum}

Define
\[
K_\xi = 2(C_\lambda+\rho^2C_w).
\]
Then
\begin{equation}
\sum_{g\in A(k)} \frac{M}{m_g\mu_g}\|\xi_g^k\|^2
\le
\left(
\sum_{g\in A(k)} \frac{M}{m_g\mu_g}
\right)
K_\xi\tau_{\max}^2.
\end{equation}
Since \(A(k)\subseteq\{1,\dots,G\}\), this sum is uniformly bounded over all active sets. Define
\begin{equation}
C_{\mathrm{delay}}
=
K_\xi \cdot
\sup_{A(k)\subseteq\{1,\dots,G\}}
\left(
\sum_{g\in A(k)} \frac{M}{m_g\mu_g}
\right).
\end{equation}
Then
\begin{equation}
\sum_{g\in A(k)} \frac{M}{m_g\mu_g}\|\xi_g^k\|^2
\le
C_{\mathrm{delay}}\tau_{\max}^2.
\label{eq:xi_bound}
\end{equation}

\subsection{Final Convergence Theorem}

We are now ready to state the main result.

\begin{theorem}[Sublinear Convergence Rate with Bounded Delay]
\label{thm:convergence}
Suppose that Assumptions~\ref{assu:strong_convex}--\ref{assu:bounded_delay} hold. If the delay is bounded by \(\tau_{\max}\), then the algorithm achieves an \(\mathcal{O}(1/T)\) sublinear convergence rate to an error neighborhood. Specifically,
\begin{equation}
\frac{1}{T}\sum_{k=0}^{T-1}\sum_{g\in A(k)}
\mathbb{E}\big[\|w_g^{k+1}-w^*\|^2\big]
\le
\frac{\mathcal{O}(1)}{T}
+
\mathcal{O}(\tau_{\max}^2).
\end{equation}
Moreover, as \(T\to\infty\),
\begin{equation}
\lim_{T\to\infty}
\frac{1}{T}\sum_{k=0}^{T-1}\sum_{g\in A(k)}
\mathbb{E}\big[\|w_g^{k+1}-w^*\|^2\big]
\le
\mathcal{O}(\tau_{\max}^2).
\end{equation}
\end{theorem}

\begin{IEEEproof}
Summing \eqref{eq:descent_ineq} from \(k=0\) to \(T-1\) gives
\begin{equation}
\begin{alignedat}{1}
&\sum_{k=0}^{T-1}\mathbb{E}[V^{k+1}-V^k] \nonumber\\
&\le
-
\sum_{k=0}^{T-1}
\mathbb{E}
\Bigg[
\sum_{g\in A(k)}
\left(
\frac{m_g\mu_g}{M}-\rho
\right)
\|w_g^{k+1}-w^*\|^2
\Bigg] \nonumber\\
&\quad
+
\sum_{k=0}^{T-1}
\mathbb{E}
\Bigg[
\sum_{g\in A(k)}
\frac{M}{m_g\mu_g}\|\xi_g^k\|^2
\Bigg].
\end{alignedat}
\end{equation}
The left-hand side telescopes to \(\mathbb{E}[V^T]-V^0\). Using \eqref{eq:xi_bound}, we obtain
\begin{equation}
\begin{split}
\mathbb{E}[V^T]-V^0
&\le
-
\sum_{k=0}^{T-1}
\mathbb{E}
\bigg[
\sum_{g\in A(k)}
\left(
\frac{m_g\mu_g}{M}-\rho
\right)
\|w_g^{k+1}-w^*\|^2
\bigg] \\
&\quad
+
T C_{\mathrm{delay}}\tau_{\max}^2.
\end{split}
\end{equation}
Rearranging yields
\begin{equation}
\begin{alignedat}{1}
&\sum_{k=0}^{T-1}
\mathbb{E}
\bigg[
\sum_{g\in A(k)}
\left(
\frac{m_g\mu_g}{M}-\rho
\right)
\|w_g^{k+1}-w^*\|^2
\bigg] \nonumber\\
&\le
V^0-\mathbb{E}[V^T]+T C_{\mathrm{delay}}\tau_{\max}^2.
\end{alignedat}
\end{equation}
Since \(V^k\ge 0\), we have
\[
V^0-\mathbb{E}[V^T]\le V^0.
\]
Define
\[
C_{\mathrm{coeff}}
=
\min_g\left(\frac{m_g\mu_g}{M}-\rho\right)>0.
\]
Then
\begin{equation}
\begin{split}
\frac{1}{T}\sum_{k=0}^{T-1}
\mathbb{E}
\bigg[
\sum_{g\in A(k)} \|w_g^{k+1}-w^*\|^2
\bigg]
\le
\frac{V^0}{T C_{\mathrm{coeff}}}
+
\frac{C_{\mathrm{delay}}}{C_{\mathrm{coeff}}}\tau_{\max}^2.
\end{split}
\end{equation}
The first term gives the \(\mathcal{O}(1/T)\) rate. Letting \(T\to\infty\) removes this transient term and leaves the \(\mathcal{O}(\tau_{\max}^2)\) error neighborhood.
\end{IEEEproof}

\noindent\textit{Remark on exact minimization and stochastic implementation.}
The analysis above is established for an idealized deterministic version of P-GADMM, where each group subproblem is solved exactly. This allows the effect of bounded delay to be written explicitly. In the implemented algorithm, the group update is approximated by local SGD and edge aggregation, and therefore includes stochastic gradient noise. Under the standard bounded variance assumption, this approximation introduces an additional residual term in the final error bound. This term is additive and does not change the qualitative dependence of the bound on \(\tau_{\max}\). Therefore, the delay effect is still characterized by the same term involving \(\tau_{\max}\), while stochastic approximation introduces an additional error floor.

\section{Experiments}
\label{sec:experiments}

This section evaluates the implementation of P-GADMM in heterogeneous edge networks. The experiments are conducted using a discrete-event simulation framework implemented in PyTorch. In the simulator, group updates are produced by synchronized local SGD within each active group, followed by edge aggregation. Unlike evaluations that only count communication rounds, the simulator advances according to event completion time. Therefore, local computation delay, waiting overhead, and straggler blocking are included in the reported wall-clock time. This setting allows the comparison to reflect both optimization behavior and system-level delay.

\subsection{Experimental Setup}
\label{subsec:setup}

\textbf{(a) Experimental platform and learning tasks:}
The experiments are conducted on a Windows 11 workstation equipped with an Intel Core Ultra 7 265K CPU at 3.9 GHz and an NVIDIA GeForce RTX 5060 Ti GPU. The software environment is Python 3.10.20 with PyTorch 2.10.0+cu130. Two image classification tasks are considered, namely MNIST and CIFAR10. For MNIST, we use a lightweight convolutional neural network with two convolutional layers and two fully connected layers. The channel dimensions increase from 1 to 16 and from 16 to 32, and each convolutional layer is followed by a ReLU activation and a $2\times2$ max pooling layer. For CIFAR10, we use a deeper convolutional network with four convolutional blocks. The channel dimensions increase as $3\!\to\!64\!\to\!128\!\to\!256\!\to\!512$, followed by a classifier with dropout.

\textbf{(b) Heterogeneity settings and baselines:}
Both statistical heterogeneity and system heterogeneity are considered. Statistical heterogeneity is controlled by a Dirichlet distribution with parameter $\alpha$. We set $\alpha=100$ for a nearly homogeneous setting and $\alpha=0.1$ for a strongly heterogeneous setting. System heterogeneity is controlled by a Pareto distribution with shape parameter $a$. We set $a=2.0$ for a light-tailed delay profile and $a=1.1$ for a heavy-tailed delay profile with stronger straggler effects. Here, the IID setting denotes a nearly homogeneous data partition induced by a large Dirichlet parameter. For convenience, the setting $(\alpha=100, a=2.0)$ is referred to as the IID setting, and the setting $(\alpha=0.1, a=1.1)$ is referred to as the non-IID setting.

P-GADMM is compared with two representative ADMM baselines. Asynch-ADMM follows the asynchronous distributed ADMM method in \cite{Chang2016Async}, where the master updates after receiving information from only a subset of workers under a bounded delay condition. GADMM follows the group-based ADMM method in \cite{Wang2017GADMM}, where a group layer is introduced between workers and the master, and the global variable is updated through group variables rather than directly through all local variables. These two baselines represent two common ways to improve the efficiency of distributed ADMM \cite{Boyd2011}. Asynch-ADMM relaxes full synchronization, whereas GADMM reduces the number of variables involved in each synchronization step. By contrast, P-GADMM combines computation-aware grouping with bounded asynchronous updates.

For a controlled comparison, all methods use the same dataset, model architecture, data partition, client computation profile, batch size, local epoch number, and random seeds under each setting. The same discrete-event simulator is used for all methods, so the reported wall-clock time is computed from the same local computation and waiting-time model. For grouped methods, the group size is kept identical unless otherwise stated. GADMM uses the same number of groups as P-GADMM, but does not use estimated training latency for grouping. Therefore, the comparison separates the effect of computation-aware grouping from the effect of using a grouped ADMM structure.

\textbf{(c) Implementation details and evaluation protocol:}
Unless otherwise stated, the number of clients is 50, the group size is 10, the learning rate is 0.01, the number of local epochs is 1, and the batch size is 64. For P-GADMM, the staleness threshold is set to $\tau_{\max}=1$. The penalty parameter is set to $\rho=1.0$ in the IID setting and $\rho=0.001$ in the non-IID setting. For Asynch-ADMM, the parameters are set to $\rho=500$, $\tau_{\max}=20$, and $\gamma=0$. These algorithm-specific parameters are kept fixed across repeated runs under the same heterogeneity setting.

Each setting is repeated three times with different random seeds. The same set of seeds is used for all methods. Performance is evaluated using five metrics: the number of rounds required to reach a target accuracy, the average waiting time per round, the straggler blocking ratio, the final test accuracy, and the wall-clock time to reach the target accuracy. The line plots show representative convergence trajectories under the same seed for all methods, whereas the tables report the mean $\pm$ standard deviation over the repeated runs.

If a method reaches the target accuracy in all repeated runs, the reported time is the time to target. If a method fails to reach the target accuracy within the maximum simulation budget, the reported time corresponds to the total simulation time, and the rounds entry is marked accordingly. When only part of the repeated runs reach the target accuracy, the table notes the number of successful runs. This convention is used to distinguish successful convergence to the target from runs that terminate at the simulation budget.

\subsection{Performance Comparison}
\label{subsec:main_results}

We first compare P-GADMM with the two baselines under both the IID and non-IID settings. To ensure a controlled comparison, all methods use the same data partition, latency model, network architecture, and optimization setting. The only difference lies in the coordination mechanism.

\begin{table*}[!t]
\centering
\caption{Experimental results on MNIST under IID and non-IID settings.}
\label{tab:mnist_stats}
\footnotesize
\setlength{\tabcolsep}{6pt}
\renewcommand{\arraystretch}{1.15}
\begin{threeparttable}
\begin{tabular}{lccccc}
\toprule
\multicolumn{6}{c}{\textbf{IID setting} ($\alpha=100$, $a=2.0$, target accuracy $95\%$)} \\
\midrule
Method & Target rounds & Avg. wait/round (s) & Blocking (\%) & Final acc. (\%) & Time to target (s) \\
\midrule
Asynch-ADMM & $206.7 \pm 66.0$ & $0.1446 \pm 0.1165$ & $8.33 \pm 2.36$ & $96.48 \pm 0.39$ & $26.2449 \pm 17.0137$ \\
GADMM       & $68.3 \pm 6.2$ & $1.6681 \pm 1.0538$ & $100.00 \pm 0.00$ & $97.39 \pm 0.51$ & $117.4762 \pm 82.5656$ \\
P-GADMM     & $113.3 \pm 9.4$ & $0.0350 \pm 0.0017$ & $13.60 \pm 5.56$ & $98.00 \pm 0.26$ & $3.9937 \pm 0.1216$ \\
\midrule
\multicolumn{6}{c}{\textbf{non-IID setting} ($\alpha=0.1$, $a=1.1$, target accuracy $80\%$)} \\
\midrule
Method & Target rounds & Avg. wait/round (s) & Blocking (\%) & Final acc. (\%) & Time to target (s) \\
\midrule
Asynch-ADMM & $90.0^{\dagger}$ & $0.3657 \pm 0.1029$ & $8.33 \pm 2.36$ & $83.70 \pm 3.79$ & $25.2723 \pm 3.7539$ \\
GADMM       & $>500^{\ddagger}$ & $4.4898 \pm 0.7215$ & $100.00 \pm 0.00$ & $51.15 \pm 6.80$ & $2244.9162 \pm 360.7430$ \\
P-GADMM     & $53.3 \pm 8.5$ & $0.0468 \pm 0.0049$ & $24.27 \pm 7.72$ & $94.77 \pm 0.27$ & $2.5739 \pm 0.5158$ \\
\bottomrule
\end{tabular}
\begin{tablenotes}[flushleft]
\item Mean $\pm$ standard deviation over three runs.
\item[$\dagger$] Target reached in 2 of 3 runs.
\item[$\ddagger$] Target not reached within 500 rounds; time is the total simulation time.
\end{tablenotes}
\end{threeparttable}
\end{table*}

\begin{figure}[!htbp]
    \centering
    \begin{minipage}[t]{0.49\columnwidth}
        \centering
        \includegraphics[width=\linewidth,trim=8 4 6 4,clip]{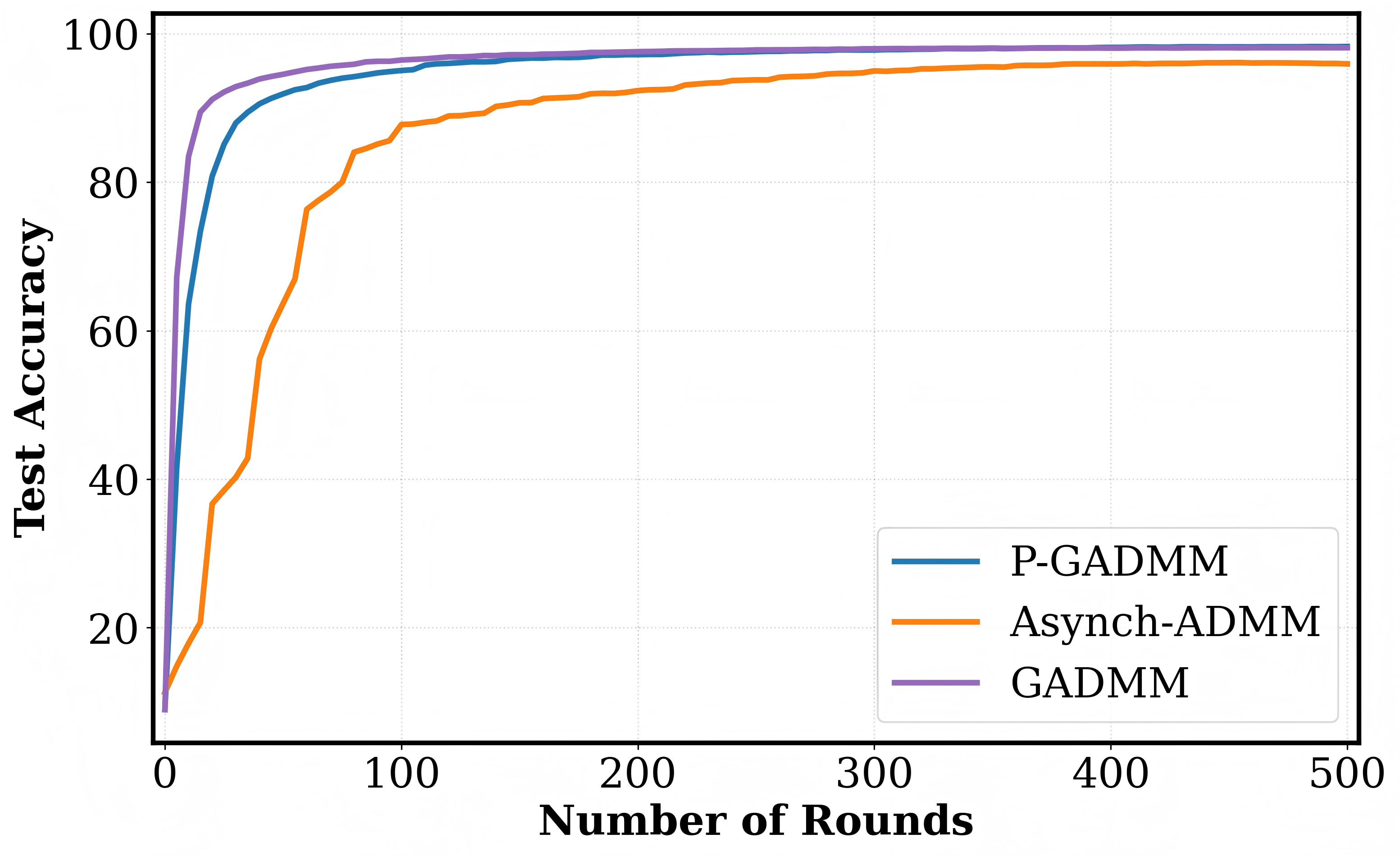}
        \centerline{\footnotesize (a)}
    \end{minipage}
    \hfill
    \begin{minipage}[t]{0.49\columnwidth}
        \centering
        \includegraphics[width=\linewidth,trim=8 4 6 4,clip]{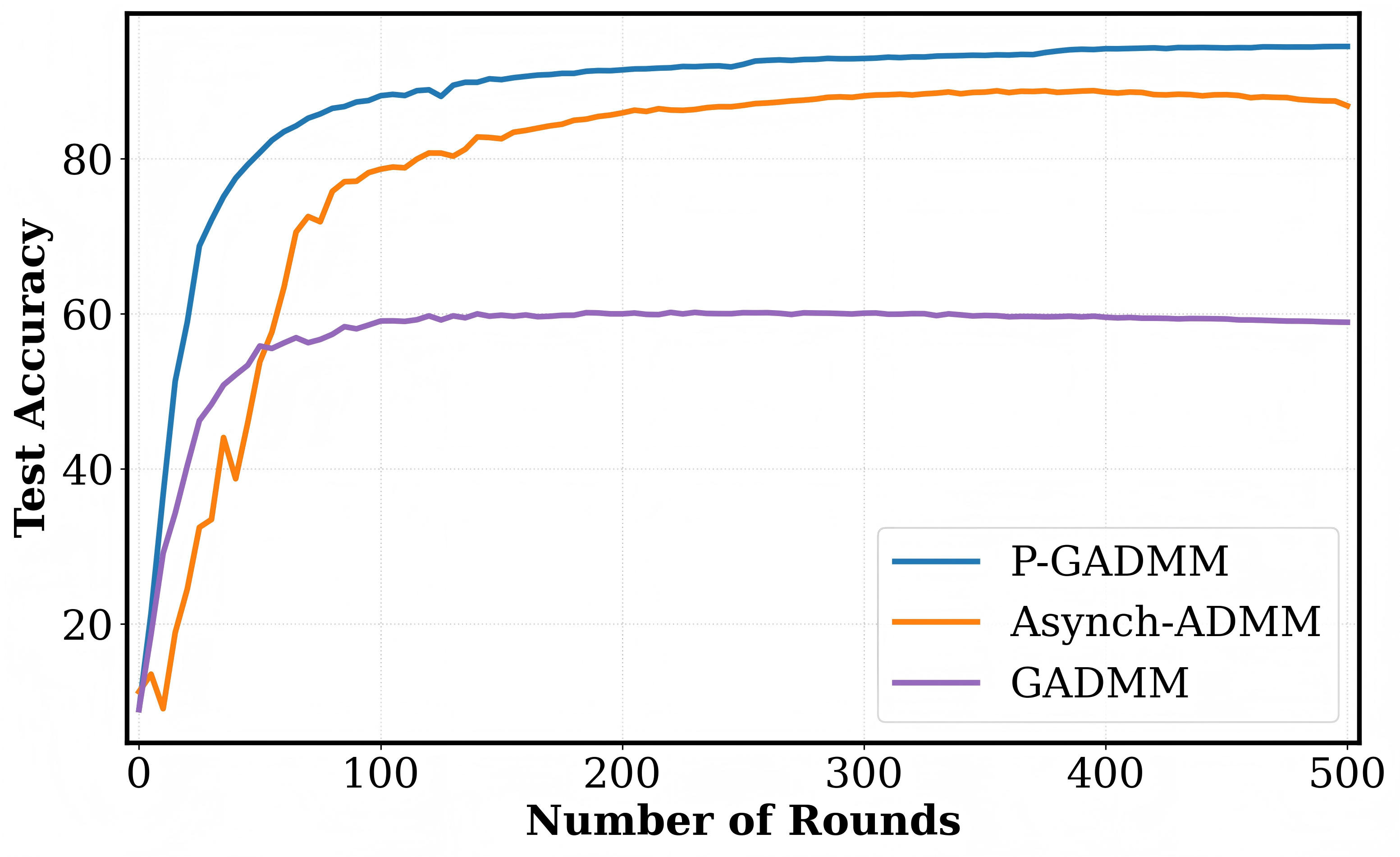}
        \centerline{\footnotesize (b)}
    \end{minipage}
    \caption{Convergence performance on MNIST. (a) IID setting. (b) non-IID setting.}
    \label{fig:mnist_rounds}
\end{figure}

The performance comparison on MNIST is shown in Fig.~\ref{fig:mnist_rounds} and Table~\ref{tab:mnist_stats}. From the reported results, several observations can be made. Under the IID setting, GADMM reaches the target accuracy in the fewest rounds, but its average waiting time per round is much larger than that of the other methods, and its blocking ratio stays at $100.00\%$. As a result, its time to target is much longer than that of P-GADMM. By comparison, P-GADMM reaches the target accuracy in $113.3 \pm 9.4$ rounds and $3.9937 \pm 0.1216$~s, while also achieving the highest final accuracy, $98.00 \pm 0.26\%$. Asynch-ADMM requires $206.7 \pm 66.0$ rounds and $26.2449 \pm 17.0137$~s. These results indicate that, even when the round count of GADMM appears smaller, the practical efficiency of P-GADMM is still better because the waiting overhead is much lower.

Under the non-IID setting, the advantage of P-GADMM becomes more pronounced. P-GADMM reaches the target accuracy of $80\%$ in $53.3 \pm 8.5$ rounds and $2.5739 \pm 0.5158$~s. Asynch-ADMM reaches the same target in only 2 of 3 runs and requires $25.2723 \pm 3.7539$~s, while GADMM does not reach the target within 500 rounds. In addition, P-GADMM achieves the highest final accuracy, $94.77 \pm 0.27\%$. Therefore, on MNIST, P-GADMM provides lower wall-clock training time while maintaining comparable or higher final accuracy in these settings.

\begin{table*}[!t]
\centering
\caption{Experimental results on CIFAR10 under IID and non-IID settings.}
\label{tab:cifar_stats}
\footnotesize
\setlength{\tabcolsep}{6pt}
\renewcommand{\arraystretch}{1.15}
\begin{threeparttable}
\begin{tabular}{lccccc}
\toprule
\multicolumn{6}{c}{\textbf{IID setting} ($\alpha=100$, $a=2.0$, target accuracy $75\%$)} \\
\midrule
Method & Target rounds & Avg. wait/round (s) & Blocking (\%) & Final acc. (\%) & Time to target (s) \\
\midrule
Asynch-ADMM & $990.0 \pm 289.9$ & $0.0488 \pm 0.0038$ & $8.33 \pm 2.36$ & $79.44 \pm 1.75$ & $49.3355 \pm 18.4516$ \\
GADMM       & $>2000^{\dagger}$ & $0.6674 \pm 0.2900$ & $100.00 \pm 0.00$ & $52.19 \pm 1.64$ & $1334.8499 \pm 580.0759$ \\
P-GADMM     & $886.7 \pm 198.7$ & $0.0350 \pm 0.0021$ & $10.92 \pm 5.13$ & $81.27 \pm 1.62$ & $31.4334 \pm 8.8425$ \\
\midrule
\multicolumn{6}{c}{\textbf{non-IID setting} ($\alpha=0.1$, $a=1.1$, target accuracy $60\%$)} \\
\midrule
Method & Target rounds & Avg. wait/round (s) & Blocking (\%) & Final acc. (\%) & Time to target (s) \\
\midrule
Asynch-ADMM & $963.3 \pm 324.6$ & $0.2545 \pm 0.0155$ & $8.33 \pm 2.36$ & $67.91 \pm 3.69$ & $243.7668 \pm 81.0247$ \\
GADMM       & $>2000^{\dagger}$ & $3.3781 \pm 1.1568$ & $100.00 \pm 0.00$ & $23.14 \pm 3.89$ & $6756.2943 \pm 2313.6207$ \\
P-GADMM     & $1070.0 \pm 215.2$ & $0.0475 \pm 0.0025$ & $13.10 \pm 3.75$ & $66.69 \pm 3.66$ & $50.8783 \pm 11.1394$ \\
\bottomrule
\end{tabular}
\begin{tablenotes}[flushleft]
\item Mean $\pm$ standard deviation over three runs.
\item[$\dagger$] Target not reached within 2000 rounds; time is the total simulation time.
\end{tablenotes}
\end{threeparttable}
\end{table*}

\begin{figure}[!htbp]
    \centering
    \begin{minipage}[t]{0.49\columnwidth}
        \centering
        \includegraphics[width=\linewidth,trim=8 4 6 4,clip]{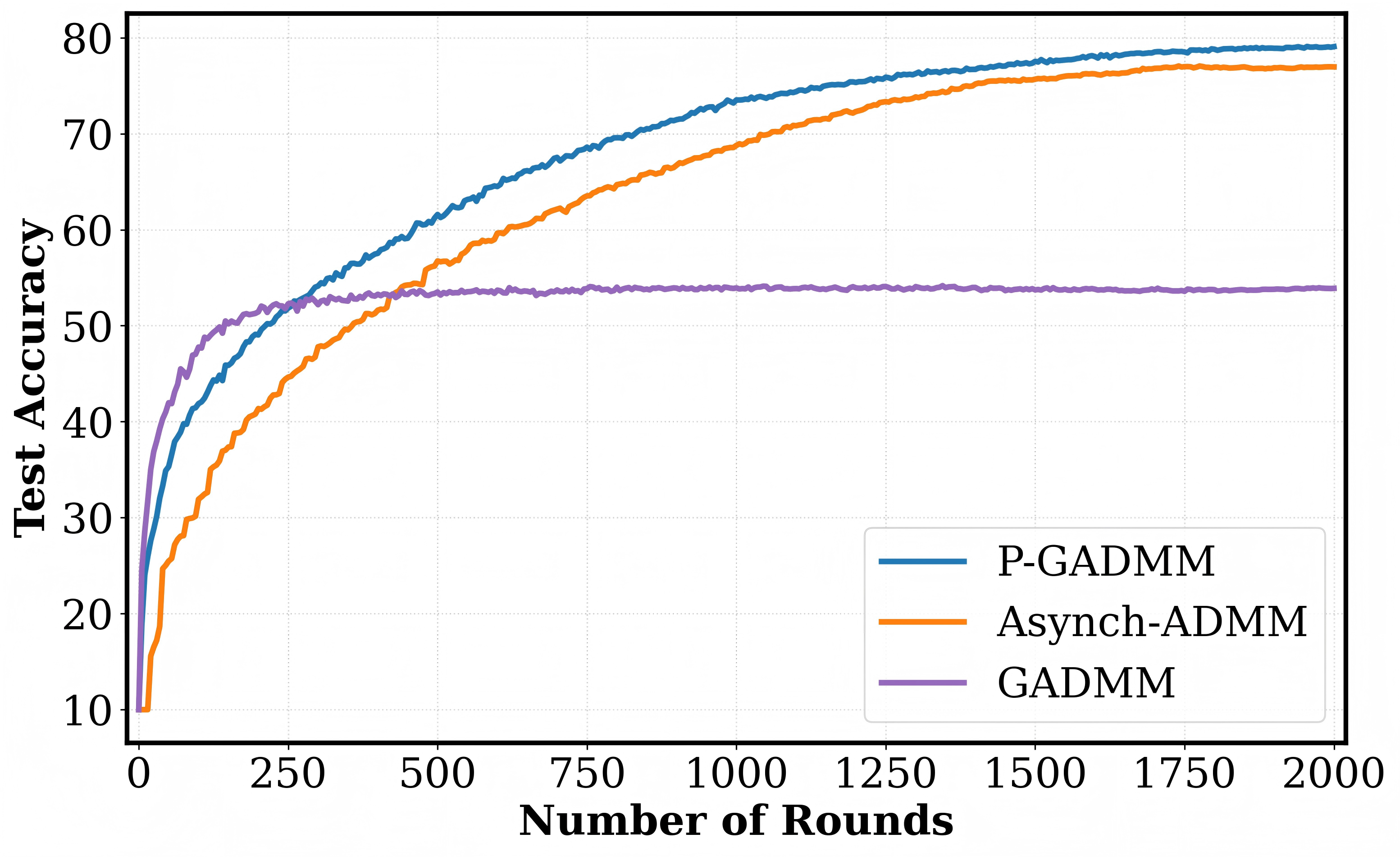}
        \centerline{\footnotesize (a)}
    \end{minipage}
    \hfill
    \begin{minipage}[t]{0.49\columnwidth}
        \centering
        \includegraphics[width=\linewidth,trim=8 4 6 4,clip]{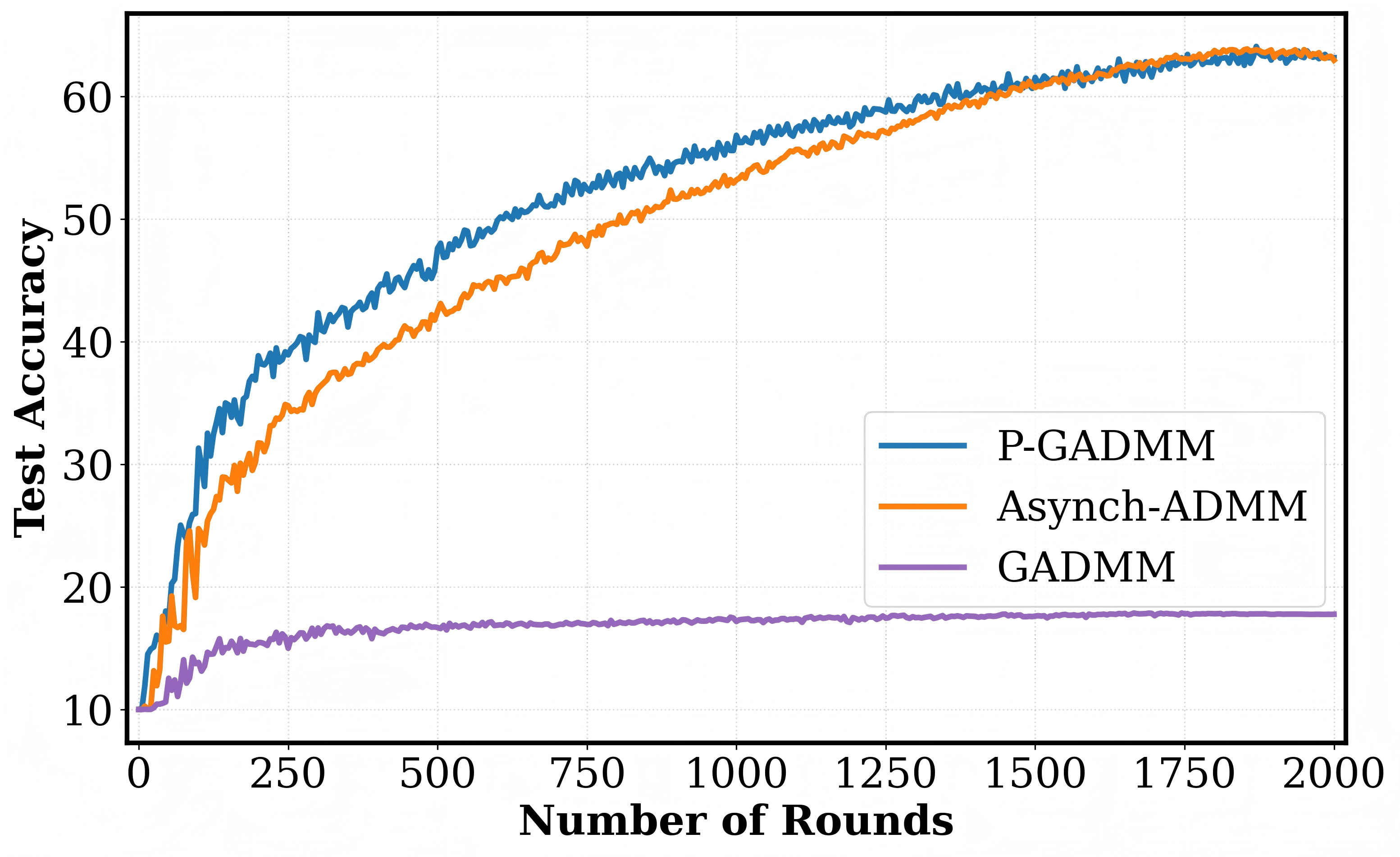}
        \centerline{\footnotesize (b)}
    \end{minipage}
    \caption{Convergence performance on CIFAR10. (a) IID setting. (b) non-IID setting.}
    \label{fig:cifar_rounds}
\end{figure}

\begin{figure*}[!t]
    \centering
    \begin{minipage}[t]{0.32\textwidth}
        \centering
        \includegraphics[width=\linewidth,trim=8 4 6 4,clip]{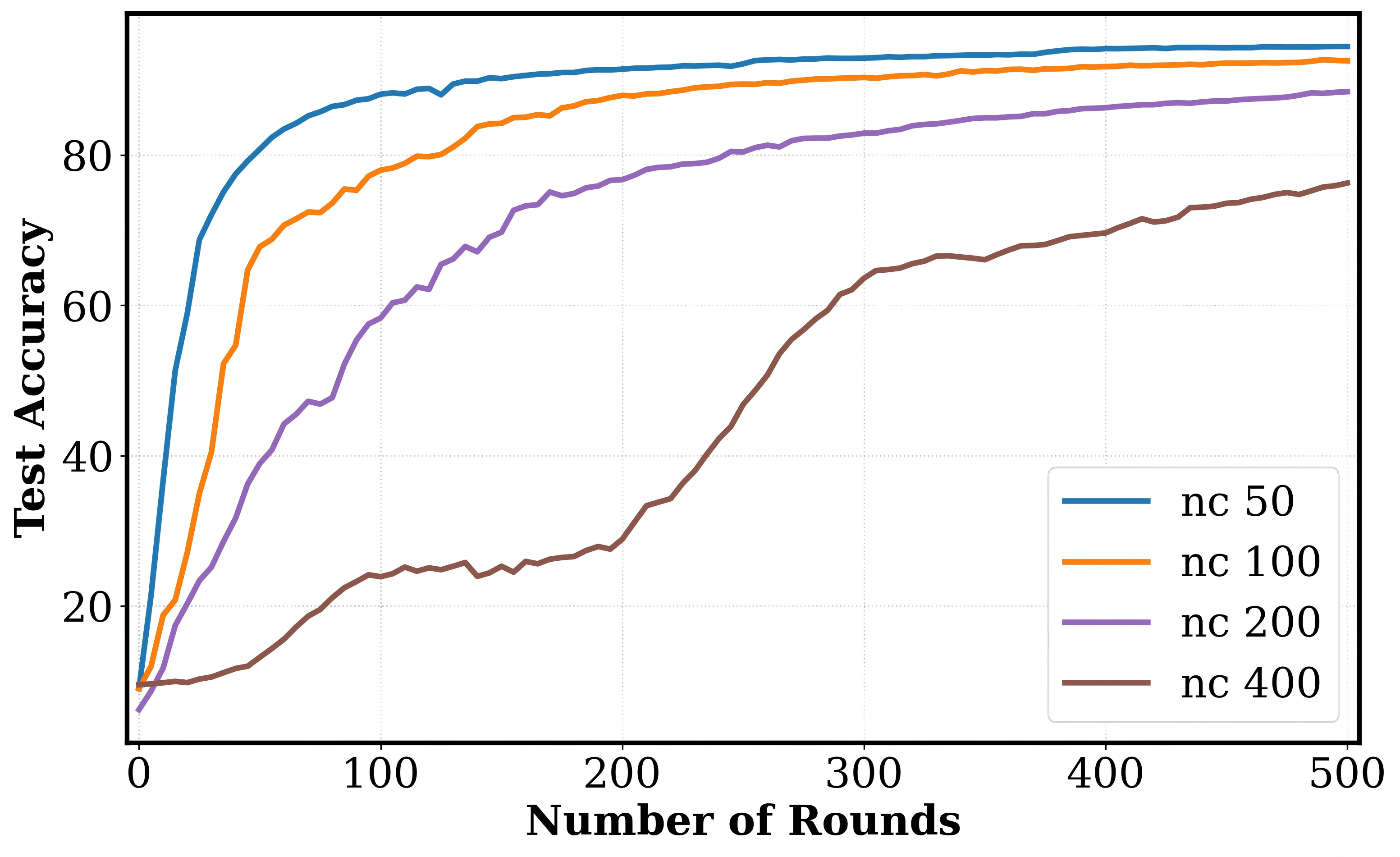}
        \centerline{\footnotesize (a)}
    \end{minipage}
    \hfill
    \begin{minipage}[t]{0.32\textwidth}
        \centering
        \includegraphics[width=\linewidth,trim=8 4 6 4,clip]{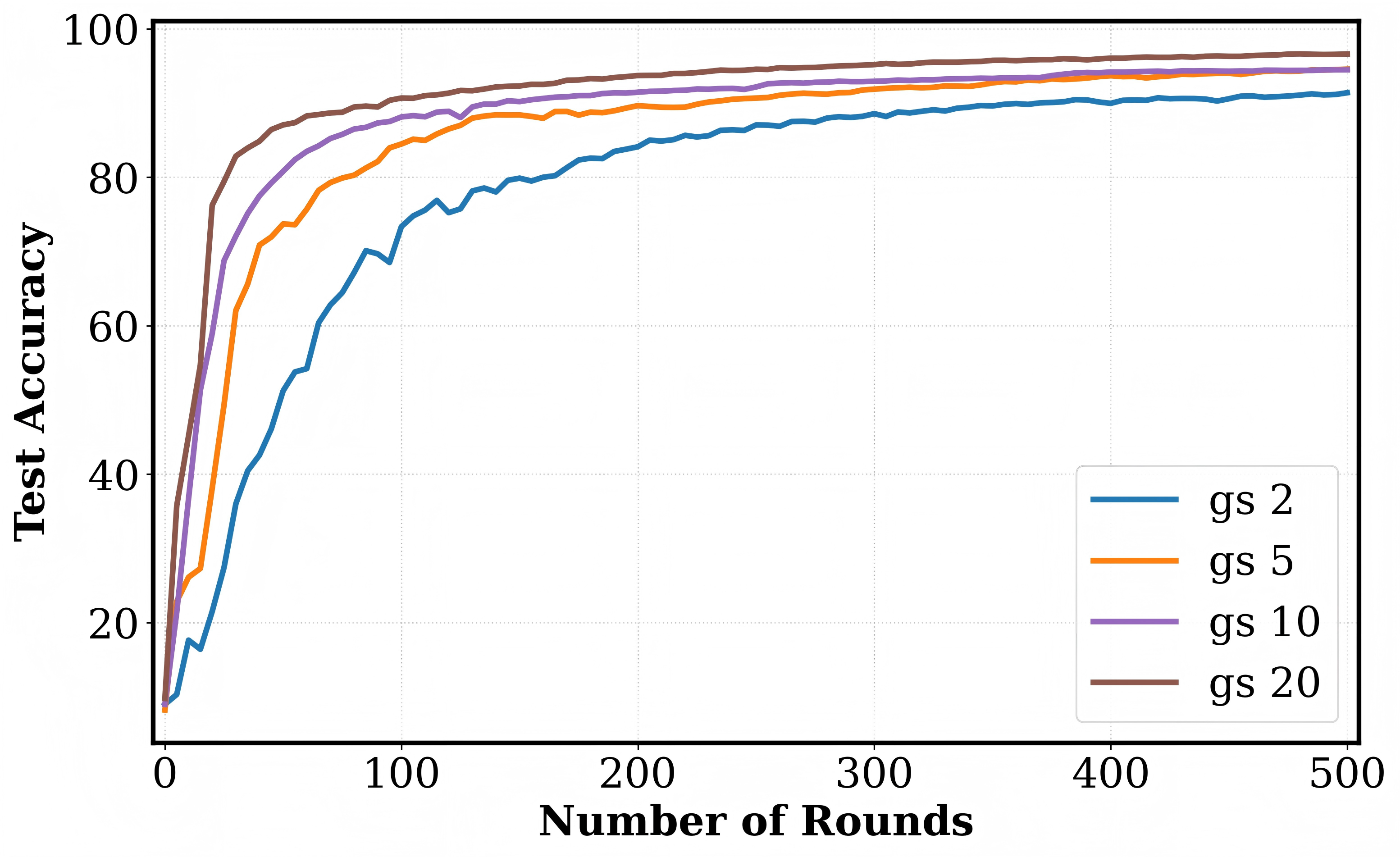}
        \centerline{\footnotesize (b)}
    \end{minipage}
    \hfill
    \begin{minipage}[t]{0.32\textwidth}
        \centering
        \includegraphics[width=\linewidth,trim=8 4 6 4,clip]{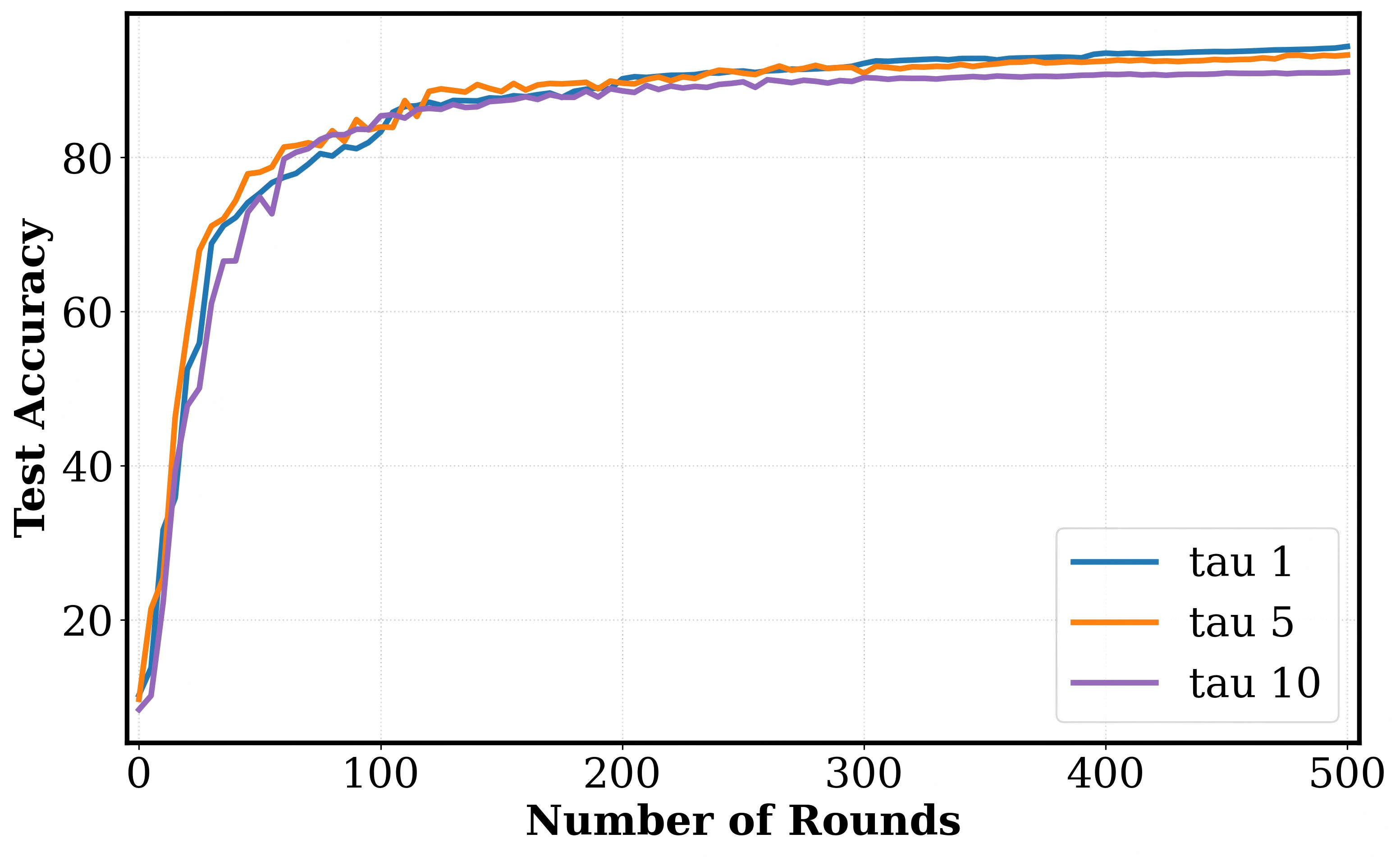}
        \centerline{\footnotesize (c)}
    \end{minipage}
    \caption{Sensitivity analysis of P-GADMM on MNIST under the non-IID setting. (a) Client population. (b) Group size. (c) Staleness threshold.}
    \label{fig:mnist_sensitivity}
\end{figure*}

\begin{figure*}[!t]
    \centering
    \begin{minipage}[t]{0.32\textwidth}
        \centering
        \includegraphics[width=\linewidth,trim=8 4 6 4,clip]{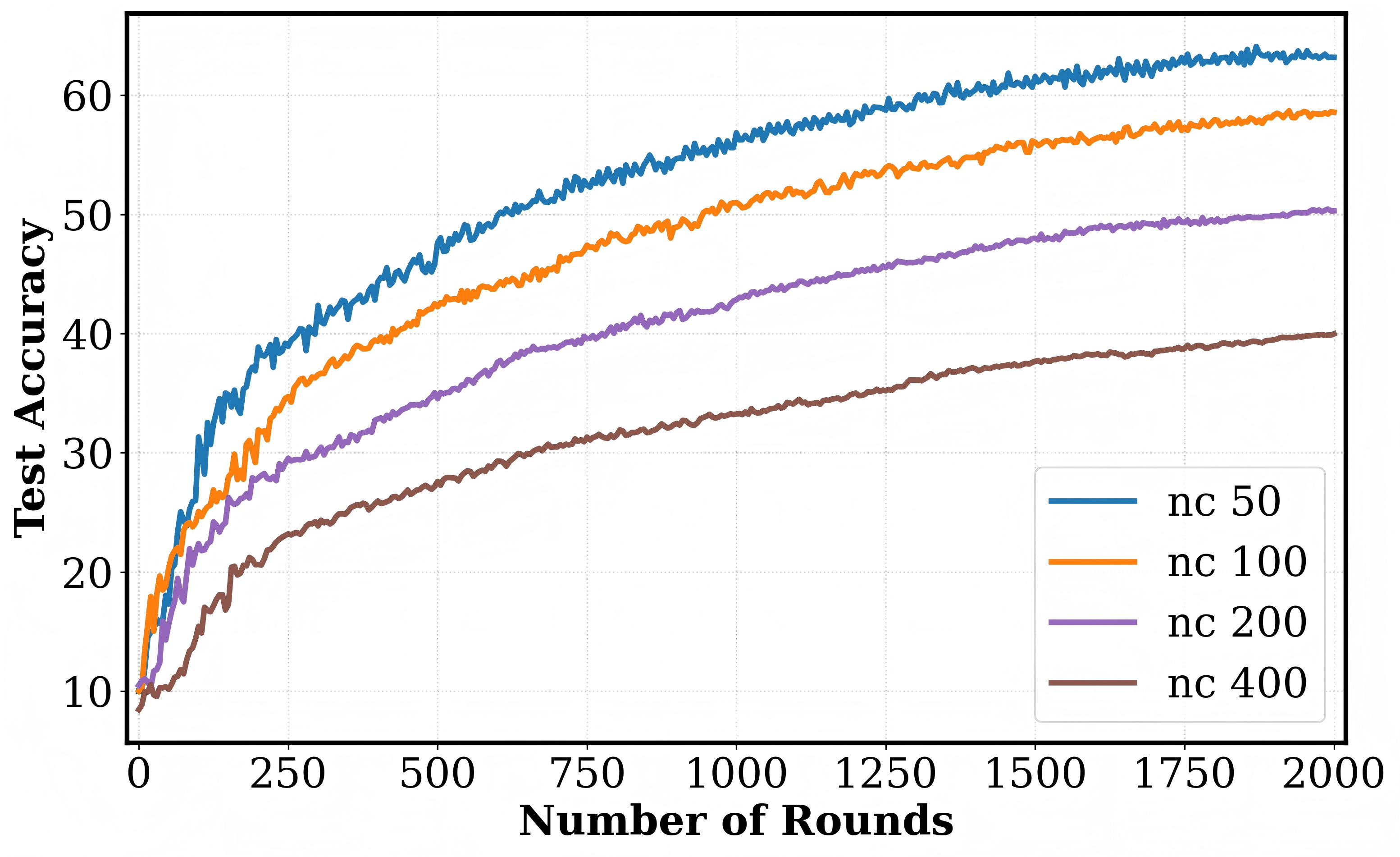}
        \centerline{\footnotesize (a)}
    \end{minipage}
    \hfill
    \begin{minipage}[t]{0.32\textwidth}
        \centering
        \includegraphics[width=\linewidth,trim=8 4 6 4,clip]{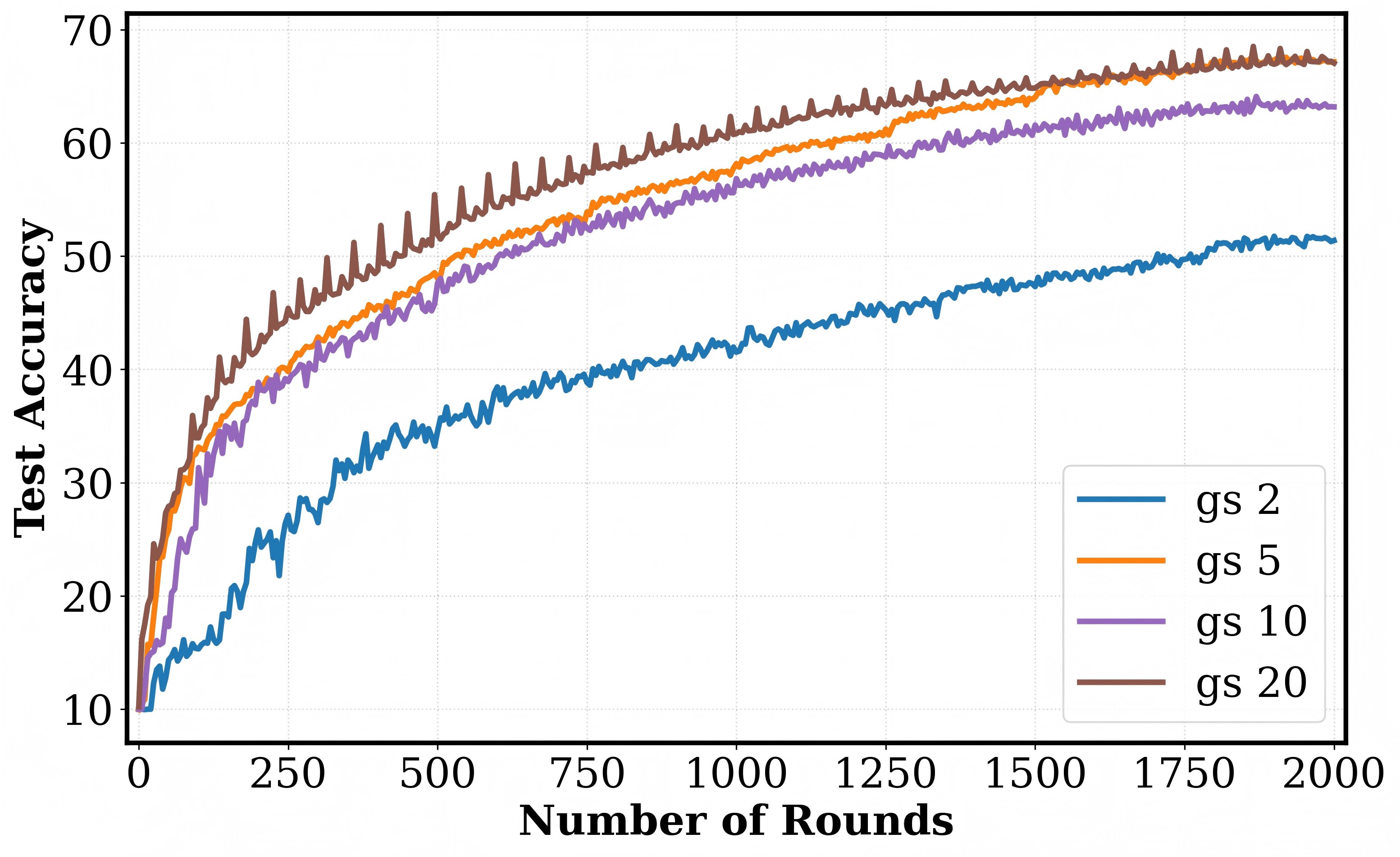}
        \centerline{\footnotesize (b)}
    \end{minipage}
    \hfill
    \begin{minipage}[t]{0.32\textwidth}
        \centering
        \includegraphics[width=\linewidth,trim=8 4 6 4,clip]{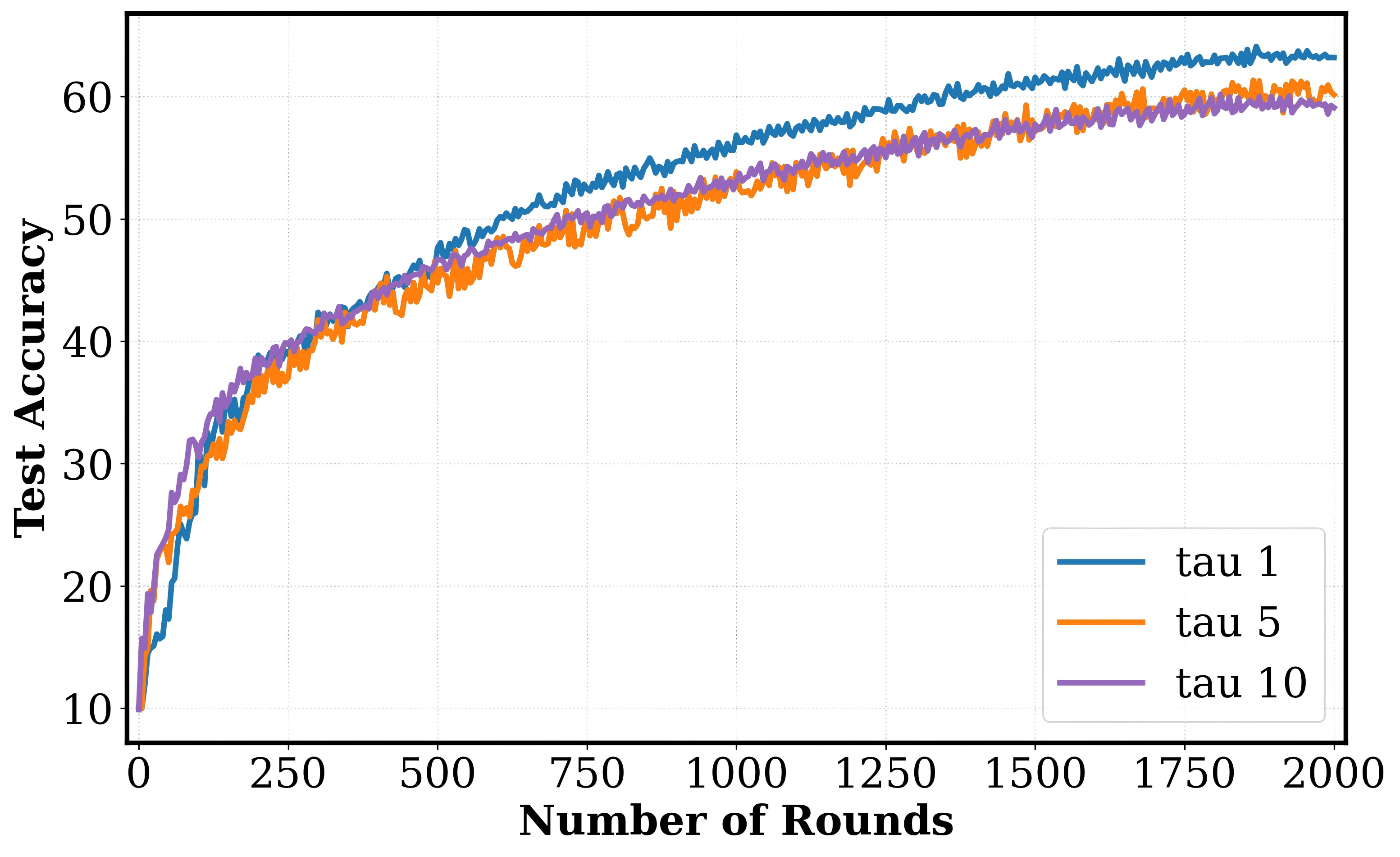}
        \centerline{\footnotesize (c)}
    \end{minipage}
    \caption{Sensitivity analysis of P-GADMM on CIFAR10 under the non-IID setting. (a) Client population. (b) Group size. (c) Staleness threshold.}
    \label{fig:cifar_sensitivity}
\end{figure*}
The performance comparison on CIFAR10 is shown in Fig.~\ref{fig:cifar_rounds} and Table~\ref{tab:cifar_stats}. Under the IID setting, P-GADMM reaches the target accuracy of $75\%$ in $886.7 \pm 198.7$ rounds and $31.4334 \pm 8.8425$~s. Asynch-ADMM requires more rounds and more time to target, while GADMM does not reach the target within 2000 rounds and ends with much lower final accuracy. In this case, P-GADMM not only achieves the shortest time to target, but also reaches the highest final accuracy among the three methods.

Under the non-IID setting, both Asynch-ADMM and P-GADMM reach the target accuracy of $60\%$. However, the difference in time to target remains large. Asynch-ADMM requires $243.7668 \pm 81.0247$~s, whereas P-GADMM requires only $50.8783 \pm 11.1394$~s. Although the final accuracy of P-GADMM, $66.69 \pm 3.66\%$, is slightly lower than that of Asynch-ADMM, $67.91 \pm 3.69\%$, the reduction in training time is still substantial. GADMM again fails to reach the target and remains far behind in both accuracy and time. These results suggest that, on the more challenging CIFAR10 task, the main strength of P-GADMM is its ability to reduce actual training time while maintaining competitive learning performance.

From the results on MNIST and CIFAR10, a consistent trend can be observed. First, P-GADMM does not always use the fewest rounds, but it consistently reduces actual training time. Second, the repeated runs show that the method remains stable across different random seeds. Third, the comparison with the two baselines clarifies the role of the proposed design. GADMM reduces the number of variables involved in synchronization, but its progress is still constrained by slow groups. Asynch-ADMM relaxes synchronization, but its advantage becomes weaker when stale updates accumulate. By combining latency-aware grouping with bounded asynchronous updates, P-GADMM achieves a better balance between efficiency and model quality.

\subsection{Sensitivity Analysis}
\label{subsec:sensitivity}

We next study the effects of client population, group size, and staleness threshold under the non-IID setting. To ensure a controlled comparison, the heterogeneity level is fixed and only one factor is changed at a time. The corresponding results are shown in Fig.~\ref{fig:mnist_sensitivity}, Fig.~\ref{fig:cifar_sensitivity}, and Table~\ref{tab:tau}. These results further illustrate how the main design parameters affect the behavior of P-GADMM.

We first examine the effect of client population. Fig.~\ref{fig:mnist_sensitivity}(a) and Fig.~\ref{fig:cifar_sensitivity}(a) show that convergence generally improves as the number of clients increases. When more clients participate, more groups can produce updates in parallel, and the cloud server can receive aggregated information more frequently. This trend is observed on both datasets. Therefore, the proposed framework can benefit from increased parallelism when more client resources are available.

We then examine the effect of group size. Fig.~\ref{fig:mnist_sensitivity}(b) and Fig.~\ref{fig:cifar_sensitivity}(b) show that larger group sizes lead to slower convergence. When the total number of clients is fixed, increasing the group size reduces the number of groups, which lowers the frequency of group updates received by the cloud. In addition, a larger group is more likely to include slow clients, which increases the chance of delayed group updates. This result is consistent with the design motivation of computation-aware grouping.

\begin{table}[!htbp]
\centering
\caption{Sensitivity analysis with different $\tau_{\max}$ under the non-IID setting.}
\label{tab:tau}
\footnotesize
\setlength{\tabcolsep}{8pt}
\renewcommand{\arraystretch}{1.12}
\begin{threeparttable}
\begin{tabular}{lccc}
\toprule
Dataset & $\tau_{\max}$ & Final acc. (\%) & Time (s) \\
\midrule
\multirow{3}{*}{MNIST}
& 1  & $94.75 \pm 0.37$ & $25.23 \pm 0.49$ \\
& 5  & $92.01 \pm 1.84$ & $23.59 \pm 3.32$ \\
& 10 & $92.08 \pm 0.85$ & $22.34 \pm 2.16$ \\
\midrule
\multirow{3}{*}{CIFAR10}
& 1  & $66.55 \pm 4.48$ & $94.92 \pm 6.20$ \\
& 5  & $60.15 \pm 1.43$ & $93.63 \pm 9.46$ \\
& 10 & $57.38 \pm 1.51$ & $96.52 \pm 12.65$ \\
\bottomrule
\end{tabular}
\end{threeparttable}
\end{table}

We finally examine the effect of the staleness threshold. Table~\ref{tab:tau}, together with Fig.~\ref{fig:mnist_sensitivity}(c) and Fig.~\ref{fig:cifar_sensitivity}(c), shows the results under repeated runs. From the table and figures, two observations can be made. First, on both MNIST and CIFAR10, the highest final accuracy is achieved when $\tau_{\max}=1$. On MNIST, the final accuracy reaches $94.75 \pm 0.37\%$, while on CIFAR10 it reaches $66.55 \pm 4.48\%$. Second, although the simulation time changes only slightly across different thresholds, the accuracy generally decreases as $\tau_{\max}$ increases. These results indicate that, in our setting, a small staleness threshold is preferable. In particular, $\tau_{\max}=1$ provides the best accuracy without causing a clear increase in training time.

\section{Conclusion}
\label{sec:conclusion}

In this paper, we proposed P-GADMM for distributed optimization in heterogeneous edge networks. The method combines computation-aware grouping, edge-level aggregation, and bounded asynchronous coordination. By grouping clients according to their estimated training latency, P-GADMM reduces the computation-speed variation within each group. The bounded asynchronous update further allows active groups to participate in global aggregation without waiting for slower groups. For strongly convex objectives, we established convergence guarantees for an idealized form of P-GADMM and characterized the effects of bounded staleness on the convergence behavior. Experiments on MNIST and CIFAR10 showed that P-GADMM reduces wall-clock training time compared with representative baselines while maintaining comparable final accuracy under both IID and non-IID settings. Future work will extend the analysis to the stochastic implementation and study adaptive grouping and coordination under dynamic network conditions.

\end{document}